\documentclass[12pt]{article}
\usepackage{amsmath,amsfonts,amssymb, amsthm, graphicx}
\usepackage{caption}
\usepackage{subcaption}
\usepackage{multirow}
\usepackage{booktabs}
\usepackage[round]{natbib}
\usepackage[linesnumbered,ruled,vlined]{algorithm2e}
\newcommand{\ra}[1]{\renewcommand{\arraystretch}{#1}}

\newcommand{\blind}{1}

\newtheorem{proposition}{Proposition}
\newtheorem{theorem}{Theorem}
\newtheorem{definition}{Definition}
\newtheorem{corollary}{Corollary}

\newtheorem{lemma}{Lemma}

\DeclareMathOperator*{\argmin}{\mathrm{argmin}}

\DeclareMathOperator{\betah}{\hat{\beta}}

\DeclareMathOperator{\E}{\mathcal{E}}

\begin{document}

\def\spacingset#1{\renewcommand{\baselinestretch}%
{#1}\small\normalsize} \spacingset{1}

\if1\blind
{
  \title{\bf Within Group Variable Selection through the Exclusive Lasso}
  \author{Frederick Campbell\\
    Department of Statistics, Rice University\\
    and \\
    Genevera Allen \\
    Department of Statistics, Rice University,\\
    Department of Electrical Engineering, Rice University,\\
    Department of Pediatrics-Neurology, Baylor College of Medicine,\\
    Jan and Dan Duncan Neurological Research\\
     Institute, Texas ChildrenÕs Hospital
    }
  \maketitle
} \fi

\if0\blind
{
  \bigskip
  \bigskip
  \bigskip
  \begin{center}
    {\LARGE\bf Within Group Variable Selection through the Exclusive Lasso}
\end{center}
  \medskip
} \fi

\bigskip
\begin{abstract}
Many data sets consist of variables with an inherent group structure. The problem of group selection has been well studied, but in this paper, we seek to do the opposite: our goal is to select at least one variable from each group in the context of predictive regression modeling. This problem is NP-hard, but we propose the tightest convex relaxation: a composite penalty that is a combination of the $\ell_1$ and $\ell_2$ norms. Our so-called Exclusive Lasso method performs structured variable selection by ensuring that at least one variable is selected from each group. We study our method's statistical properties and develop computationally scalable algorithms for fitting the Exclusive Lasso. We study the effectiveness of our method via simulations as well as using NMR spectroscopy data. Here, we use the Exclusive Lasso to select the appropriate chemical shift from a dictionary of possible chemical shifts for each molecule in the biological sample. 
\end{abstract}

\noindent%
{\it Keywords:}  Structured Variable Selection, Composite Penalty, NMR Spectroscopy, Exclusive Lasso
\vfill

\newpage
\spacingset{1.45} 
\section{Introduction}
\label{sec:intro}

In regression problems with a predefined group structure, we seek to accurately predict the response using a subset of variables composed of at least one variable from each predefined group. We can phrase this structured variable selection problem as a constrained optimization problem where we minimize a regression loss function subject to a constraint that ensures sparsity and selects at least one variable from every predefined group. This problem has potential applications in many areas including genetics, chemistry, computer science, and proteomics. Consider a motivating example from finance. In portfolio selection, the variance of the portfolio is just as important as the expected performance of the returns \citep{markowitz1952portfolio}. Suppose we want to select an index fund comprised of a diverse set of 50 stocks whose performance approximates the performance of the S\&P 500.  We can ensure that we are selecting a diversified portfolio by requiring that we select at least one stock from every financial sector; selecting securities from different sectors diversifies the index fund and effectively lowers the variance of the return of our portfolio. We can phrase this strategy as a structured variable selection problem where we minimize the difference in performance between the S\&P 500 and our portfolio subject to selecting a small set of securities that is comprised of at least one security from each predefined financial sector.

Even though this problem is known to be NP-hard, a popular approach in the literature uses convex penalties to relax similar combinatorial problems into tractable convex problems. While the Lasso \citep{tibshirani1996regression} is the most well known of these convex relaxations, there are several frameworks specifically designed to find convex alternatives to complicated structured combinatorial problems \citep{obozinski2012convex,halabi2014totally}. These frameworks lead to convex penalties like the Group Lasso \citep{yuan2006model}, Composite Absolute Penalties \citep{zhao2009composite}, and the Exclusive Lasso \citep{zhou2010exclusive}, the subject of this paper. \citet{zhou2010exclusive} first uses the Exclusive Lasso penalty in the context of multitask learning, and \citet{obozinski2012convex} and \citet{halabi2014totally} relate the penalty to their framework for relaxing combinatorial problems.  The Exclusive Lasso penalty has not yet been explored statistically or developed into a method that can be used for sparse regression and within group variable selection. We will develop the Exclusive Lasso method and study its statistical properties in this paper.

To motivate our statistical investigation of the Exclusive Lasso for sparse regression further, consider the problem of selecting one variable per group using existing techniques such as the Lasso or Marginal Regression. If the Lasso's incoherence condition and beta-min condition are satisfied and Marginal Regression's faithfulness assumption is satisfied, then both methods recover the correct variables with out any knowledge of the group structure \citep{genovese2012comparison,wainwright2009sharp}. However, data rarely satisfies these assumptions. Consider that if two variables are correlated with each other, the Lasso often selects one instead of both variables. When whole groups are correlated, the Lasso may only select variables in one group as opposed to variables across multiple groups. Similarly, if the variables most correlated with the response are in the same group, Marginal Regression will ignore the true variables in other groups.  If we recall the portfolio selection example, we group variables together {\it because} they are correlated. In these situations, the fact that the Lasso and Marginal regression are agnostic to the group structure hurts their ability to select a reasonable set of variables across all predefined groups. If we know that this group structure is inherent to our problem, then complex real world correlated data motivate the development of new structured variable selection methods that directly enforce the desired selection across groups.  

In this paper, we investigate the statistical properties of the Exclusive Lasso for sparse, within group variable selection in regression problems. Specifically, our novel contributions beyond the existing literature \citep{zhou2010exclusive,obozinski2012convex,halabi2014totally} include:  characterizing the Exclusive Lasso solution and relating this solution to the existing statistics literature on penalized regression (Section 2); proving consistency and prediction consistency (Section 3); developing a fast algorithm with convergence guarantees for estimation (Section 4); deriving the degrees of freedom that can be used for model selection (Section 5);  and investigating the empirical performance of our method through simulations (Sections 6 and 7).

\section{The Exclusive Lasso}
\label{sec:exlasso}
Consider the linear model where the response is a linear combination of the variables subject to Gaussian noise: $y = X \beta^* + \epsilon$ where $\epsilon$ is i.i.d Gaussian. For notational convenience, we assume the response is centered to eliminate an intercept term. We assume $\beta^*$ is structured such that its indices are divided into non-overlapping, predefined, groups and that the support of $\beta^*$ is distributed across all groups.  We allow the support set within a group to be as small as one element and as large as the entire group. We can write this as two structural assumptions; (1)  there exists a collection of non-overlapping predefined groups denoted, $\mathcal{G}$, such that $\underset{g\in \mathcal{G}}{\cup}g = \{1,\dots, p\}, \underset{g\in \mathcal{G}}{\cap}=\emptyset$ and (2) the support set $S$ of the true parameter $\beta^*$ is non-empty in each group such that for all $g \in \mathcal{G}$ we have $S \cap g \ne \emptyset$ and $\beta^*_i \ne 0$ for all $i\in S$. Let $C = \{\beta \in \mathbb{R}^p, : \beta_S \neq 0, S\cap g \neq \emptyset, \forall g \in \mathcal{G}\}$ be the set of all parameters that satisfy our structural assumptions.

Our goal is to find the element in $C$ that best represents $y$ using the optimization problem: $\betah = \underset{\beta \in C}{\text{argmin}} \|y-X\beta\|_2^2 $. Our constraint set makes this a combinatorial problem and is generally NP- hard. Instead of considering the problem as stated, we study its convex relaxation by replacing the combinatorial constraint with the convex penalty $P(\beta) = \frac{1}{2}\underset{g \in \mathcal{G}}{\sum} \|\beta_g\|_1^2$ first proposed in the context of document classification and multitask-learning \citep{zhou2010exclusive}. \citet{obozinski2012convex} showed that the Exclusive Lasso penalty is in fact the tightest convex relaxation for the combinatorial constraint requiring the solution to contain exactly one variable from each group.

In this paper we propose to study the Exclusive Lasso penalty in the context of penalized regression, looking at both the constrained version:

\begin{equation}
\betah = \underset{\beta}{\text{argmin }}\frac{1}{2}\|y -X\beta \|_2^2\ \text{ subject to } P(\beta) \le \tau
\end{equation}

\noindent where $\tau$ is some positive constant and its lagrangian
\begin{equation}
\betah = \underset{\beta}{\text{argmin}}\frac{1}{2}\|y -X\beta \|_2^2 + \lambda  \frac{1}{2}\underset{g \in \mathcal{G}}{\sum} \|\beta_g\|_1^2 
\end{equation}

\noindent We predominantly work with the Lagrangian as they are equivalent because it is a convex problem.

Now let us understand the penalty better. For each group $g$, the penalty takes the 1-norm of the parameter vector restricted to the group $g$, $\beta_g$, and then take the 2-norm of the vector of norms. If each element is its own group, the penalty is equivalent to ridge regression. If all elements are in the same group, the penalty is equivalent to squaring the 1-norm penalty. Loosely, the penalty performs selection within group by applying separate lasso penalties to each group. At the group level, the penalty is a ridge penalty preventing entire groups from going to zero. Whatever the case, the group structure informs the type of regularization because it is a composite penalty, utilizing the $\ell_1$ and $\ell_2$ norms within and between groups respectively.

As an illustration, consider the following toy example. Let $\beta^* = (\beta^*_{1,1}, \beta^*_{1,2}, \beta^*_{2,1})$ be our parameter such that the first index denotes group membership and the second denotes the element within group. If we evaluate the penalty at this parameter we have $2P(\beta^*) = (|\beta^*_{1,1}|+|\beta^*_{1,2}|)^2 + (\beta^*_{2,1})^2$. We can visualize this example using the Exclusive Lasso's unit ball as shown in Figure 1. Restricting our attention to variables in the same group $\beta^*_{1,1}, \beta^*_{1,2}$ and setting $\beta^*_{2,1} = 0$ yields a unit ball equivalent to the ball generated by the $\ell_1$ norm. Alternatively, if we restrict our attention to variables in different groups $\beta^*_{1,1}, \beta^*_{2,1}$ and set $\beta^*_{1,2}=0$ the unit ball is equivalent to the ball generated by the $\ell_2$ norm. The geometry of simple convex penalties dictate the structure of the estimate in constrained least squares problems \citep{chandrasekaran2012convex} suggesting that if  the $\ell_1$-norm enforces sparsity in its estimate and that the $\ell_2$-norm enforces density, we can expect the Exclusive Lasso to send either $\beta^*_{1,1}$ or $\beta^*_{1,2}$ to zero while never sending $\beta^*_{2,1}$ to zero.

\begin{figure}[h1]
        \centering
        \begin{subfigure}[b]{0.3\textwidth}
                \includegraphics[width=\textwidth]{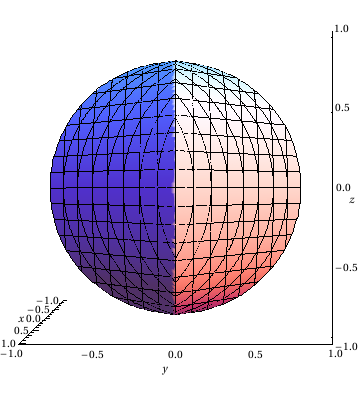}
                \caption{The unit ball is equivalent to the $\ell_2$ ball between groups, enforcing density.}
                \label{fig:b2}
        \end{subfigure}%
        ~ 
        \begin{subfigure}[b]{0.3\textwidth}
                \includegraphics[width=\textwidth]{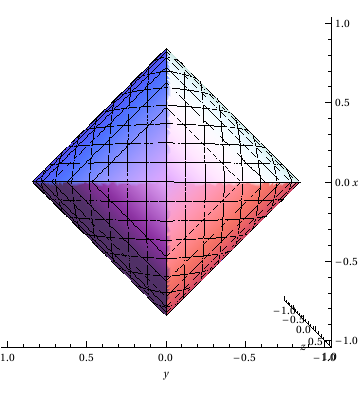}
                \caption{The unit ball is equivalent to the $\ell_1$ ball within group, enforcing sparsity.}
                \label{fig:b1}
        \end{subfigure}
        ~ 
        \begin{subfigure}[b]{0.34\textwidth}
                \includegraphics[width=\textwidth]{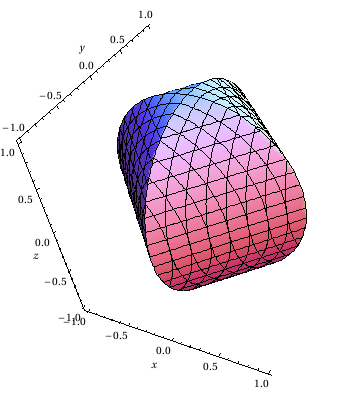}
                \caption{The Exclusive Lasso unit ball.}
                \label{fig:b3}
        \end{subfigure}
        \caption{The unit ball for the Exclusive Lasso penalty. The ball has properties of both the $\ell_1$ unit ball and the the $\ell_2$ unit ball. Let  $\beta^* = (\beta^*_{1,1}, \beta^*_{1,2}, \beta^*_{2,1})$ be a parameter with two groups where the first index denotes the group and the second index enumerates the elements within a group. Considering the perspective where $\beta^*_{2,1}=0$  yields a ball equivalent to the $\ell_1$ ball (b). Considering a perspective where either $\beta^*_{1,1} = 0$ or $\beta^*_{1,2} = 0$ yields a ball equivalent to the $\ell_2$ ball (a).}\label{fig:balls}
\end{figure}

Like the Group Lasso, studied by \citet{yuan2006model}, the Exclusive Lasso assumes the variables have an inherent group structure. However the Group Lasso also assumes that only a small number of groups represent the response $y$. Consequently, the Group Lasso penalty performs selection at the group level sending entire groups to zero. Despite their differences, both the Exclusive Lasso and the Group Lasso are examples of a broader class of penalties studied by \citet{zhao2009composite} called Composite Absolute Penalties. Composite Absolute Penalties employ combinations of $\ell_p$ norms to effectively model a known grouped or hierarchical structure. The first norm is applied to the coefficients in a group. This enforces the desired structure within group. The second norm is applied at the group level to the vector of group norms. This yields the desired structure between groups. In a sense, the Exclusive Lasso is the opposite of the Group Lasso. Where the Exclusive Lasso employs an $\ell_1$-norm within group and an $\ell_2$-norm between groups, the Group Lasso uses an $\ell_2$-norm within group and an $\ell_1$-norm between groups. Several authors have investigated some of the well known composite penalties. \citet{nardi2008asymptotic} study the conditions under which the Group Lasso correctly identifies the correct support. \citet{negahban2008joint} study the theoretical properties of the $\ell_1 / \ell_{\infty}$ norm penalty, a penalty similar to the Group Lasso. Despite the work on other composite penalties, the statistical properties of Exclusive Lasso have not yet been studied.
		
\subsection{Optimality Conditions}
 
 We use the first order optimality conditions to characterize the active set and derive two expressions for the Exclusive Lasso estimate $\betah$. Each of these expressions offers insight into either the behavior of the estimate or its statistical properties.
 
Because problem (1) is convex, an optimal point satisfies $-X^T(y-X\betah) +\lambda z = 0$ where $z$ is an element of the sub gradient such that
  \begin{equation}
  z_i \in 
\partial P(\betah)=
  \begin{cases}
   sign(\betah_i)\|\betah_g\|_1 & \text{if } \betah_i \ne 0 , i\in g \\
   \left[- \|\betah_g\|_1, \  \|\betah_g\|_1\right]       & \text{if } \betah_i= 0 , i \in g
  \end{cases}
\end{equation}

Alternatively, we can express the sub gradient as the product of a matrix and a vector. If we let $M_g  = sign(\betah_{s\cap g})sign(\betah_{s\cap g})^T$ and let $M_S$ be a block diagonal matrix with matrices $M_g$ on the diagonal, then the sub gradient restricted to the support set $S$ of $\betah$ will be $z_S = M_S\betah_S$.

\noindent Note that the matrix $M_S$ depends on the support set as the block diagonal matrices are defined by the nonzero elements of $\betah$ in each group. 

\begin{proposition}

If $S$ is the support set of $\betah$, we can express $\betah$ in terms of the support set:

\begin{equation}
\betah_S = (X_S^TX_S + \lambda M_S )^{\dagger}X^T_Sy  \text{ and } \betah_{S^c} = 0
\end{equation}
\end{proposition}

The matrix $M_S$ distinguishes the Exclusive Lasso from similar estimates like Ridge Regression. It is a block diagonal matrix that is only equivalent to the identity matrix when there is exactly one nonzero variable in each group. At this point, the Exclusive Lasso behaves like a Ridge Regression estimate on the nonzero indices that it has selected. 

Note that this characterization describes the behavior of the nonzero variables but it does not describe the behavior of the entire active set as we vary $\lambda$. To derive a second characterization of $\betah$, we note that the optimality conditions imply that every nonzero variable in the same group has an equal correlation with the residual $X_i^T(y-X\betah)$. This allows us to determine when variables enter and exit the active set. Recall that there is always at least one nonzero variable in each group. Another variable only enters the active set once its correlation with the residual is equal to the correlation shared by the other nonzero variables in the same group. We call the set $\E = \left \{i : \frac{|X_i^T (y - X\betah)|}{\|\betah_g\|_1} = \lambda \right \}$ the ``weighted equicorrelation set" because of its resemblance to the equicorrelation set described in \citet{efron2004least}.

We can use this set to derive an explicit formula for $\betah$. 

\begin{proposition}
If $\E$ is the weighted equicorrelation set, $i$ is in group $g$, and $\gamma'$ is a vector such that $\gamma'_i = \|\betah_g\|_1 - |\betah_i|$ then,

\begin{equation}
\betah_{\E} = (X_{\E}^TX_{\E} + \lambda I )^{-1}[X^T_{\E}y-  \lambda \gamma' s]  \text{ and } \betah_{\E^c} = 0
\end{equation}
where $s \in \{-1,1\}^{|\E|}$ is a vector of signs that satisfies the optimality conditions  and $\E^c$ is the compliment of the set $\E$.
\end{proposition}

 The expression points to the general behavior of the penalty. For the non-zero indices, the first term is a ridge regression estimate $(X_{\E}^TX_{\E} + \lambda I )^{-1}X^T_{\E}y$. The second term $ (X_{\E}^TX_{\E} + \lambda I )^{-1} \lambda \gamma' s$ adaptively shrinks the variables to zero. In the case where the all groups have exactly one non-zero element the Exclusive Lasso estimate is a ridge regression estimate, ensuring that there is at least one non-zero element in each group.

This characterization also helps us see that our method is not guaranteed to estimate exactly one non-zero element in each group. Selecting exactly one element from each group depends on the response $y$ and the design matrix $X$. We believe that the degree of correlation between the columns of the design matrix impact the probability of selecting greater than one element per group. In comparison to other methods, we recover the correct structure at much higher rates, but it is possible to construct examples that prevent the Exclusive Lasso from estimating the correct structure. See the appendix for more details.

Before proceeding we use a small simulated example to compare the behavior of the Lasso to the behavior of the Exclusive Lasso. We let $y=X\beta^* + \epsilon$ where $\epsilon \sim N(0,1)$. The design matrix $X\in \mathbb{R}^{20 \times 30}$ is multivariate normal with covariance that encourages correlation between groups and within groups. The incoherence condition is not satisfied with $\||X_{S^c}^TX_S(X_S^TX_S)^{-1}|\|_{\infty} =  2.603$. There are five groups and $\beta^*$ is nonzero for one variable in each group. In Figure 2, we show the Exclusive Lasso and Lasso regularization paths for this example. In the figure the solid lines are the truly nonzero variables and each color represents a different group. The Exclusive Lasso sends variables to zero until there is exactly one nonzero variable in each group whereas the Lasso eventually sends all variables to zero. Further, notice that the Lasso does not enforce the proper structure. The first five variables to enter the regularization path only represent three of the five groups. Because of this, the Lasso misses several true variables. The regularization path also highlights the Exclusive Lasso's connection to Ridge Regression; five variables will never go to zero.

\begin{figure}
        \centering
        \begin{subfigure}[b]{0.5\textwidth}
                \includegraphics[width=\textwidth]{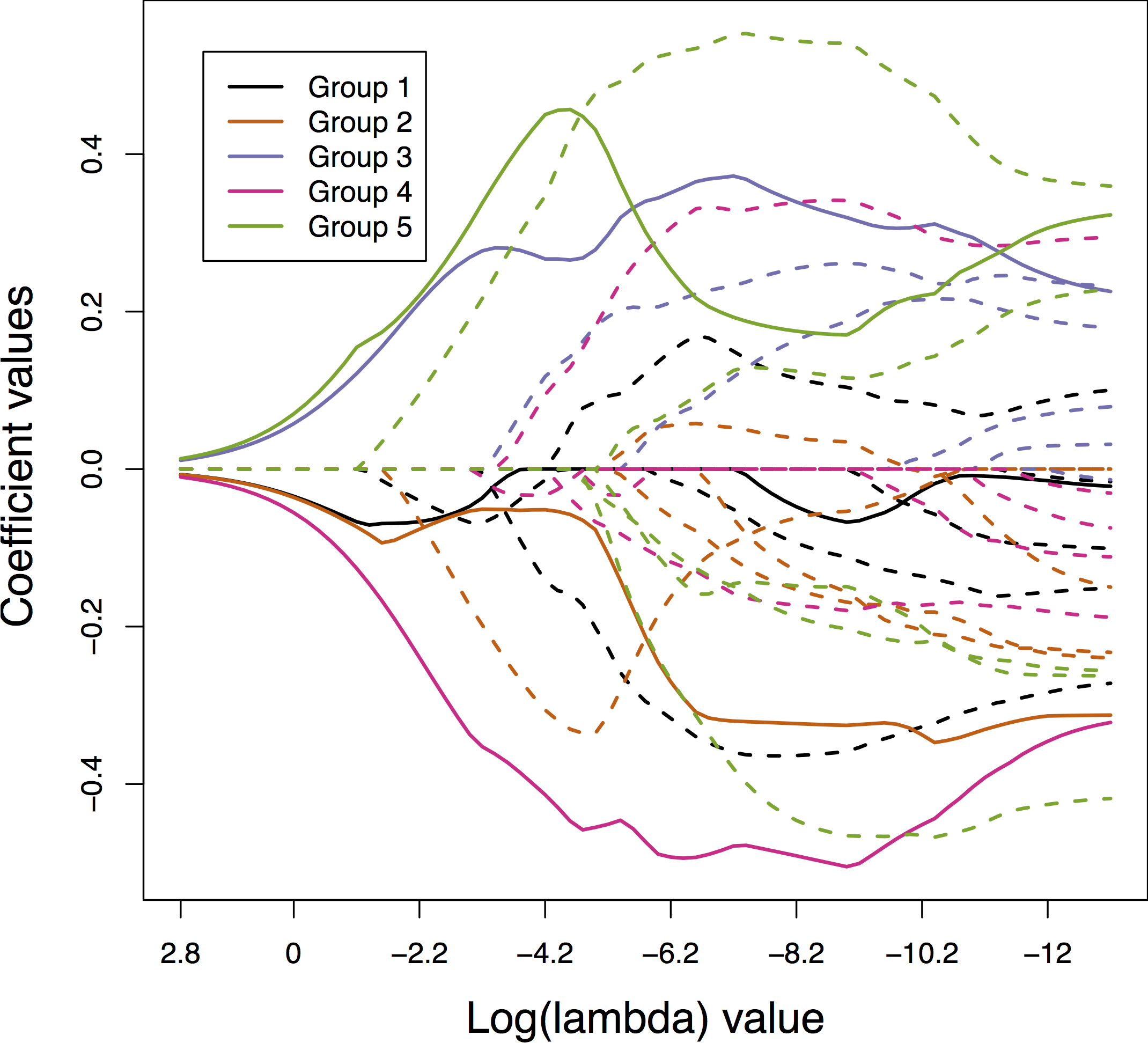}
                \caption{}
                \label{fig:all}
        \end{subfigure}%
        ~ 
        \begin{subfigure}[b]{0.5\textwidth}
                \includegraphics[width=\textwidth]{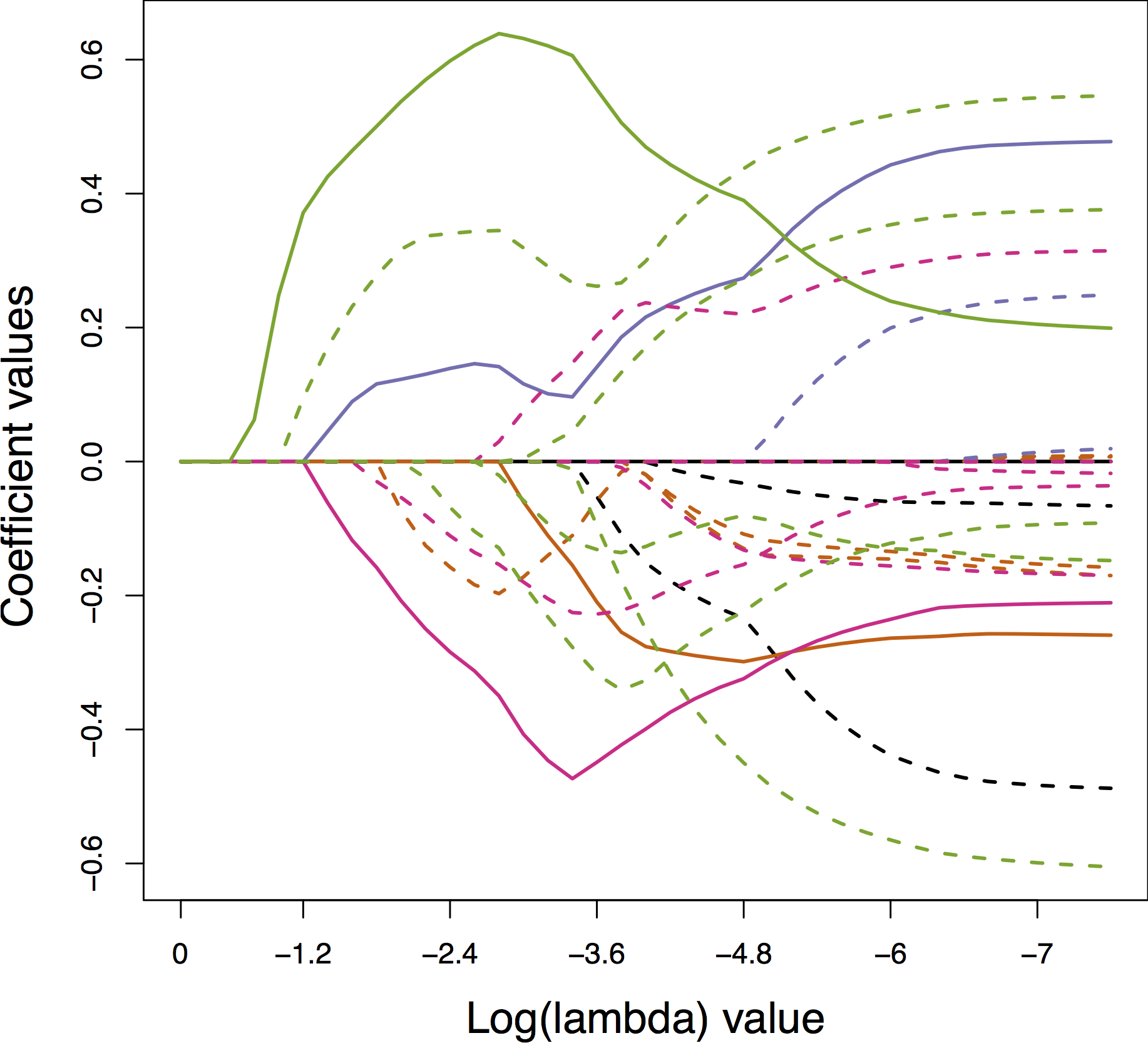}
                \caption{}
                \label{fig:lasso}
        \end{subfigure}
            \caption{A toy simulation with n = 20 and p = 30 consisting of five groups with one true variable per group. The coefficient paths of the true variables are solid and non-true variables are dashed lines. Each color represents a different group. (a) Regularization path for the Exclusive Lasso. The Exclusive Lasso behaves like an adaptively regularized Ridge Regression estimate sending variables to zero until only one variable from each group is nonzero. At this point it behaves like a Ridge Regression estimate. (b) Regularization path for the Lasso. The Lasso sends variables to zero without considering the group structure. Note that the first five variables to enter the model for the Lasso represent only groups 3, 4 and 5, where as the Exclusive Lasso has five variables, at least one from each group, that are in the model for all $\lambda$.}\label{fig:path}
\end{figure}

\section{Statistical Theory}
\label{s:theory}
The Exclusive Lasso is prediction consistent under weak assumptions. These assumptions are relatively easy to satisfy in practice compared to the assumptions typically associated with sparsistency results or consistency in the $\ell_2$-norm.  Throughout the rest of this section, we use the following notation: as before, $X \in \mathbb{R}^{n \times p}$ denotes the design matrix and $\beta^* \in \mathbb{R}^p$ is the true parameter. We let $\mathcal{G}$ be a collection of non overlapping groups such that $\underset{g \in \mathcal{G}}{\cup} = \{1,2,\dots,p\}$ and for all $g,h \in \mathcal{G}$, $g \cap h = \emptyset$. We let $S$ denote the support set of $\beta^*$, meaning that for all $ i \in S$, $\beta^*_i \neq 0$. We denote elements of $X$ as $X_{ij}$ and we index the columns of $X$ by group so that $X_g$ are the columns corresponding to group $g$. Let $Y^* = X\beta^*$ and $\hat{Y} = X\betah$ where the vector $\betah$ is the estimate produced by minimizing squared error loss subject to $P(\betah) \le K$ for some constant $K$.  The population mean squared prediction error is  $\text{MSPE}(\betah) = \mathbb{E}( Y^*- \hat{Y})^2$  and the estimated mean squared prediction error is $\widehat{\text{MSPE}}(\betah) = \frac{1}{n}\overset{n}{\underset{i = 1}{\sum}} ( Y_i^*- \hat{Y}_i)^2$. Note that we can also rewrite them so that MSPE$(\betah) = \mathbb{E}\|\betah - \beta^*|\|_{\Sigma}^2$ and $\widehat{\text{MSPE}}(\betah) = \|\betah - \beta^*\|_{\hat{\Sigma}}^2$ where $\Sigma$ is the covariance matrix of $X$. Later this allows us to compare and bound the $\ell_2$-norm coefficient error by the mean squared prediction error. 
 
In order to prove prediction consistency we need three assumptions: \\

\noindent {\bf Assumption (1):} The data $X$ is generated by a probability distribution such that the columns $\{X_1 \dots X_p\}$ have covariance $\Sigma$ and the entries of $X$ are bounded so that $|X_{ij}| \le M $.\\

\noindent {\bf Assumption (2):} The value of the penalty evaluated at the true parameter is bounded so that $\frac{1}{2}P(\beta^*) \le K $.\\

\noindent {\bf Assumption (3):} The response is generated by the linear model $Y = X\beta^* + \epsilon$ where $\epsilon \stackrel{iid}{\sim} N(0,\sigma^2).$ \\

Using assumptions $(1)-(3)$ we show that the Exclusive Lasso is prediction consistent.

\begin{theorem}
Under assumptions $(1)$, $(2)$ and $(3)$, the population mean squared prediction error of $\betah$ is bounded such that 

\begin{equation}
\text{MSPE}(\betah) \le 2(K+|\mathcal{G}|)M\sigma \sqrt{\frac{2\log(2p)}{n}} + 8(K+|\mathcal{G}|)^2M^2 \sqrt{
\frac{2p\log(2p^2)}{n}}
\end{equation}
 which goes to $0$ as $n \rightarrow \infty$.
\end{theorem}

Our assumptions are similar to those of the Lasso. Authors have shown that prediction consistency for the Lasso has assumptions that are much easier to satisfy then assumptions for other consistency results like sparsistency \citep{greenshtein2004persistence}. Like the Lasso's prediction consistency assumptions, many data sets will satisfy assumption (1). If we believe the data truely arises from a linear model then assumptions (2) and (3) will be satisfied as well. 

Theorem 1 shows that the Exclusive Lasso is consistent in terms of the norm $\|x\|_{\Sigma}$. The result differs from the prediction consistency result in \citep{chatterjee2013assumptionless} by one term. The group structure in the penalty appears in the bound as the cardinality of the collection of groups. This suggests that we can allow $n$, $p$ and the number of groups to scale together and still ensure that the estimate is prediction consistent. We use this result to justify using the Exclusive Lasso for prediction when a small number of variables are desired in each group. 

We can also bound the estimated mean squared prediction error. 

\begin{theorem}
Under assumptions $(1)$, $(2)$ and $(3)$ the estimated mean squared prediction error of $\betah$ is bounded such that 

\begin{equation}
\mathbb{E}[\widehat{\text{MSPE}}(\betah)] \le 2(K+|\mathcal{G}|)M\sigma \sqrt{\frac{2\log(2p)}{n}} 
\end{equation}
 which goes to $0$ as $n \rightarrow \infty$.

\end{theorem}

Similar to Theorem 1, the Exclusive Lasso is consistent in terms of the norm $\|x\|_{\hat{\Sigma}}$ under weak assumptions. If we add a further assumption, we can show that the Exclusive Lasso is consistent using the $\ell_2$ norm. 

\begin{corollary}
If the smallest eigenvalue of the covariance matrix $\Sigma$ is bounded below by $c>0$ then the Exclusive Lasso estimate is consistent in the $\ell_2$-norm:
\begin{equation}
\|\betah -\beta^*\|_2^2 \le \frac{2}{c} (K+|\mathcal{G}|)M\sigma \sqrt{\frac{2\log(2p)}{n}} + \frac{8}{c}(K+|\mathcal{G}|)^2M^2 \sqrt{
\frac{2p\log(2p^2)}{n}}
\end{equation}
\end{corollary}

We add another assumption to establish consistency in the $\ell_2$ norm. This requires the covariance matrix to be strictly positive definite which is much more restrictive then our previous assumptions on $\Sigma$. In general, our results for the Exclusive Lasso are comparable to the consistency results for the Lasso but differ to account for the additional structure in the penalty.

\section{Estimation}
\label{s:inf}

Many types of algorithms exist to fit sparse penalized regression models including coordinate descent, proximal gradient descent, and Alternating Direction Method of Multipliers (ADMM).  We develop our Exclusive Lasso Algorithm based on proximal gradient descent because it is well studied and known to be computationally efficient.  Roughly, this type of algorithm, popularized by \citet{beck2009fast}, proceeds by moving in the negative gradient direction of the smooth loss projected onto the set defined by the non-smooth penalty. These algorithms are easy to implement for simple penalties, because simple penalties typically have closed form proximal operators.

In our case, the proximal operator associated with the Exclusive Lasso penalty is a major challenge as there is no analytical solution. The proximal operator for the Exclusive Lasso is defined as follows:

 \begin{equation} \label{eq:2}
 prox_P(z) = \underset{\beta}{\text{argmin}} \frac{1}{2}\|\beta-z\|^2_2 + \lambda \underset{g}{\sum}\|\beta_g\|_1^2
\end{equation}

\noindent We propose an iterative algorithm to compute the proximal operator of the Exclusive Lasso penalty, prove that this algorithm converges, and prove that the proximal gradient descent algorithm based on this iterative approach converges to the global solution of the Exclusive Lasso problem. 

First, we propose an algorithm to compute the proximal operator.\

\begin{lemma} 
For proximal operator $prox_P(z)$ where $P$ is our Exlcusive Lasso penalty, if $S(z,\lambda) = sign(z)(|z|-\lambda)_+$ and $\beta_g^{-i} = (\beta_{1}^{k+1}, \dots, \beta_{j-1}^{k+1},\beta_{j+1}^k, \dots , \beta_{p}^k)$ then the coordinate wise updates are: 
\begin{equation} \label{eq:3}
\beta_{i,g}^{k+1} = S(\frac{1}{1+\lambda} z_{i,g},\frac{\lambda}{1+\lambda} \|\beta^{-i}_g\|_1).
\end{equation}
\end{lemma}

\noindent Notice that each coordinate update depends on the other coordinates in the same group. Because of this, we can implement this in parallel over the groups. At each step, instead of cyclically updating all of the coordinates we update each group in parallel by cyclically updating each coordinate in a group. If there are a large number of groups or the data is very large, this can help speed up the calculation of the proximal operator. This is important in the context of our proximal gradient descent algorithm because the proximal operator is calculated at each step of the proximal gradient descent method.
Empirically, we have observed that coordinate descent is an efficient way to calculate the proximal operator. However, we still need to prove that our algorithm converges to the correct solution.

Note that because our penalty is non-separable in $\beta$, we cannot invoke standard convergence guarantees for coordinate descent schemes without additional investigation.  Nevertheless, we can guarantee our algorithm converges and defer the proof to the appendix:

\begin{theorem}
The coordinate descent algorithm converges to the global minimum of the proximal operator optimization problem given in equation (\ref{eq:2}) .
\end{theorem}

We are now ready to derive a proximal gradient descent algorithm to estimate the Exclusive Lasso using  the coordinate descent algorithm described above.  As the negative gradient of our $\ell_2$ regression loss is $-X^T(y-X\beta)$, our proximal gradient descent update is $\beta^{k+1}  = prox_P(\beta^k - \frac{1}{L}(X^TX\beta^k - X^Ty))$, where $L= \lambda_{\max}(X^TX)$ is the Lipschitz constant for $\|y-X\beta\|_2^2$ (see appendix). Note that this step and Lipschitz constant are the same for all regression problems that use an $\ell_2$-norm loss function. Putting everything together, we give an algorithm outline for our Exclusive Lasso estimation algorithm in Algorithm 1. 

 \begin{algorithm}
\DontPrintSemicolon 
\KwIn{$\beta^0 \in \mathbb{R}^p, \epsilon \in \mathbb{R}, \delta \in \mathbb{R}$}
\KwOut{$\betah \in \mathbb{R}^p$}
\While{ $\|\beta^{k+1}-\beta^k\|> \epsilon$}{
$z_g = \beta^k_g - \frac{1}{L}(X^T_gX\beta^k - X_g^Ty)$\\
In parallel for each $g$: \\
Initialize $\tilde{\beta}_g \in \mathbb{R}^{p_g}$ \\

\While{$\|\tilde{\beta}^{t + 1}_g - \tilde{\beta}_g^t \| > \delta$}{
\For{$i \gets 1$ \textbf{to} $p_g$} {
   $\beta_{g,i}^{t+1} = S(\frac{1}{\lambda+1}z_{g,i} ,\frac{\lambda}{\lambda + 1} \|\tilde{\beta}_g^{-i}\|_1)$
 }
 }
 $\beta_g^{k+1} = \tilde{\beta}_g$
 }
\Return{$\beta$}\;
\caption{{\sc Exclusive Lasso Algorithm} to fit the Exclusive Lasso}
\label{algo:ELA}
\end{algorithm} 

Next, we prove convergence of Algorithm 1. Note that we never calculate the proximal operator exactly. Our coordinate descent algorithm solves the proximal operator optimization problem to within an arbitrarily small error. We need to ensure that the proximal gradient descent algorithm converges despite this sequence of errors $\{\epsilon_k\}$. We can show that as long as the sequence of errors converges to zero, the proximal gradient descent algorithm will converge. 

\begin{theorem}
Given objective function $f(\beta) = \frac{1}{2}\|y - X\beta\| + \lambda P(\beta)$ the sequence of iterates $\{\beta^k\}$ generated by our proximal gradient descent algorithm converges in objective function at a rate of at least $O(1/k)$ when the sequences $\{\|\epsilon_k\|\}$ and $\{\sqrt{\epsilon_k}\}$ are summable.
\end{theorem}

Overall, this particular algorithm compares well to ISTA, the proximal gradient descent algorithm for the Lasso \citep{beck2009fast}. Although computing the proximal operator is more complicated due to the structure of the penalty, the convergence rate is the same order as the convergence rate for ISTA. The fact that the iterates are easy to compute and the convergence results are competitive reinforce our empirical observations; despite the additional structure, the Exclusive Lasso Algorithm compares well to first order methods for the Lasso and other penalized regression problems.

\section{Model Selection}
\label{s:bic}

In practice, we need a data-driven method to select $\lambda$ and regulate the amount of sparsity within group. To this end, we provide an estimate of the degrees of freedom that will allow us to use BIC and EBIC approaches for model selection. Note that while other general model selection procedures like cross validation and stability selection can be employed, these do not perform well for the Exclusive Lasso. Like the Lasso, cross validation tends to overselect variables. Similarly, we observe stability selection overselect variables, possibly because the Exclusive Lasso always selects at least one variable per group. If a true variable is not in the model, it is necessary replaced by a false variable leading to artificially high probabilities of inclusion and stability scores for false variables. 

The BIC formula relies on an unbiased estimate for the degrees of freedom for the Exclusive Lasso. We leverage techniques used by \citet{stein1981estimation} and \citet{tibshirani2012degrees} to calculate the degrees of freedom, but defer the proof to the appendix. Our formula leads to an unbiased estimate for the degrees of freedom that we use for both the BIC and the EBIC. Recall that the matrix $M_S$ is a block diagonal matrix where each nonzero block $M_g$ is the outer product of the sign vector of the estimate, $M_g = sign(\betah_{S \cap g})sign(\betah_{S \cap g})^T$. This leads to our statement of the degrees of freedom for $\hat{y}$:

\begin{theorem}
For any design matrix $X$ and regularization parameter $\lambda \ge 0$, if $y$ is normally distributed, then the degrees of freedom for $X\betah$ is $df(\hat{y}) = \mathbb{E}\left[trace(X_S(X_S^TX_S + \lambda M_S)^{\dagger}X^T_S)\right]$.
\end{theorem}
An unbiased estimate of the degrees of freedom is then

\begin{equation}
\widehat{df}(\hat{y})=trace[X_S(X_S^TX_S + \lambda M_S)^{\dagger}X^T_S].
\end{equation}

To verify this result, we compare our unbiased estimate of the degrees of freedom to simulated degrees of freedom following the set up outlined in \citet{efron2004least} and \citet{zou2007degrees} . Recall that for Gaussian $y$, the formula for the degrees of freedom can be stated as $df(\hat{y}) = \underset{i=1}{\overset{n}{\sum}} cov(\hat{y}_i,y_i)/\sigma^2$. This formula points to a convenient way to simulate the degrees of freedom. We let $\beta^*$  be the true parameter and we simulate $y$,  $B$ times such that $y^b = X\beta^* + \epsilon^b$ where $\epsilon^b \stackrel{iid}{\sim} N(0,1)$. We then calculate an estimate for the covariance. Because $y$ is standard Gaussian with $\sigma^2 = 1$, the simulated degrees of freedom is $\widehat{df}(\hat{y}) = \underset{i=1}{\overset{n}{\sum}} \widehat{cov}(\hat{y}_i,y_i)/\sigma^2$ where we simulate the covariances according to $\widehat{cov}_i = \frac{1}{B}\underset{b=1}{\overset{B}{\sum}} (\hat{y}^b_i - [H^b X\beta^*]_i)(y^b_i - [X\beta^*]_i)$. Note that $H^b$ is the hat matrix for the estimate $\hat{y}^b$. In other words $\mathbb{E}[\hat{y}^b] = X_S(X_S^TX_S + \lambda M)^{\dagger}X^T_SX\beta^* = H^bX\beta^*$ (where $S$ here depends on the estimate at iteration $b$). In our simulations, we set $B=2000$ and found that empirically, our unbiased estimate of the degrees of freedom closely matches the simulated degrees of freedom (Figure 3). 

\begin{figure}[h!]

  \centering
    \includegraphics[width=0.45\textwidth]{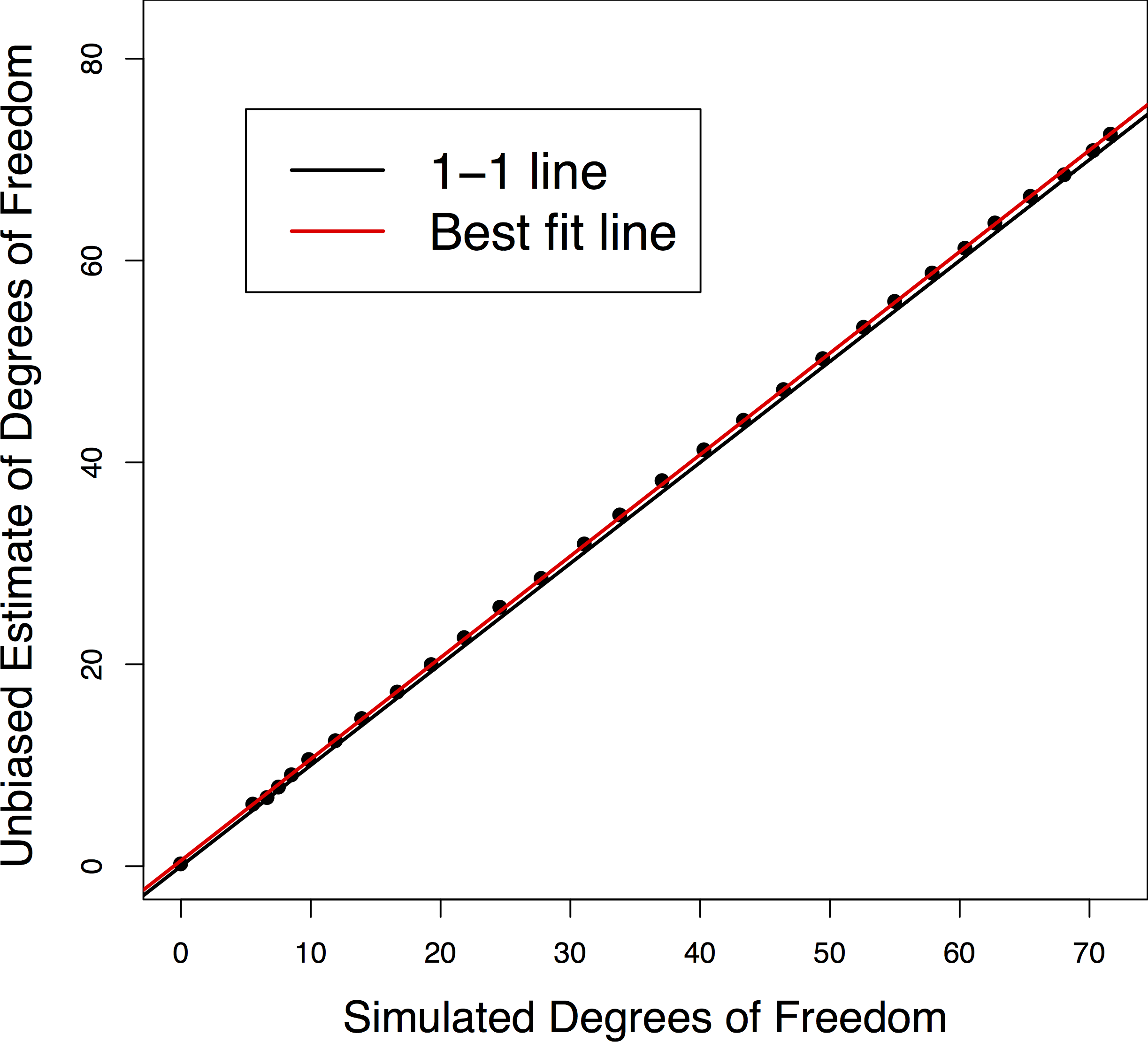}
   \caption{Comparison of our estimate for the degrees of freedom to the simulated degrees of freedom. The simulated degrees of freedom matches the estimated degrees of freedom very closely.}
\end{figure}

We can now use our unbiased estimate of the degrees of freedom to develop a model selection method for the Exclusive Lasso based on the Bayesian Information Criteria (BIC) \citep{schwarz1978estimating} and the Extended Bayesian Information Criteria (EBIC) \citep{chen2008extended}. Recall that while the BIC provides a convenient and principled method for variable selection, it can be too liberal in a high dimensional setting and is known to select too many spurious variables. \citet{chen2008extended} address this with the EBIC approach. Hence, we present both the BIC and EBIC for our method, noting that the latter is preferable in high-dimensional settings. If we assume the variance of $y$ is unknown, the respective formulas for the BIC and the EBIC are 

\begin{equation}
\text{BIC} = \log\left(\frac{\|y-\hat{y}\|^2_2}{n}\right) + \hat{\text{df}}(\hat{y}) \frac{\log(n)}{n}
\end{equation}
 
 and 
 
\begin{equation}
\text{EBIC} = \log\left(\frac{\|y - \hat{y}\|^2_2}{n}\right) + \hat{\text{df}}(\hat{y})\frac{\log(n)}{n} + \hat{\text{df}}(\hat{y})\frac{\log(p)}{n}
\end{equation}

These formulas for the BIC and the EBIC can be used to select $\lambda$ for the Exclusive Lasso in practice. Usually, we can select $\lambda$ sufficiently large to select exactly one variable per group. In cases where the design matrix does not permit selecting one variable per group, (as discussed in Sections 6 and 7) we suggest using the BIC or EBIC to select $\lambda$ and then thresholding the estimate within each group so that there is only one variable per group. We call this group-wise thresholding. 

\section{Simulation Study}
\label{sims}

We study the empirical performance of our Exclusive Lasso through two sets of simulation studies: first, for selecting one variable per group and second, for selecting a small number of variables per group. We examine three situations with moderate to large amounts of correlation between groups and within groups. We omit the low correlation setting from the simulations because they correspond to design matrices that are nearly orthogonal, satisfying both the Incoherence condition and the Faithfulness condition. This is not representative of the types of real data for which we would need to use the Exclusive Lasso and is uninteresting because all methods perform perfectly, selecting all of the truly nonzero variables and none of the false variables.

In the first simulations, we simulate data using the model $y = X\beta^* + \epsilon$ where $\epsilon \stackrel{iid}{\sim} N(0,1)$ and $\beta^*$ is the true parameter. The variables are divided into five equal sized groups and the true parameter is nonzero at one index in each group and zero otherwise. We use three design matrices each with $n=100$ observations and $p=100$ variables, to test the robustness of the Exclusive Lasso to within group correlation and between group correlation.  All three matrices are drawn from a multivariate normal distribution with a Toeplitz covariance matrix with entries $\Sigma_{ij} = w^{|i -j|}$ for variables in the same group, and $\Sigma_{ij} = b^{|i-j|}$ for variables in different groups. The first covariance matrix uses constant $b =.9$ and $w =.9$ to simulate high correlation within groups and high correlation between groups. The second covariance matrix uses $b =.6$ and $w =.9$ so that the correlation between groups is lower then the correlation within groups, resulting in high correlation within group and medium correlation between groups. The third covariance matrix uses constants $w = .6$ and $b=.6$ so that there is medium correlation both between group and within group. 

We compare two versions of our Exclusive Lasso as described in the previous section. First, we use a regularization parameter $\lambda$, large enough to ensure that the method selects exactly one element per group. In these simulations, $\lambda = \underset{i}{\max}|X_i^Ty|$ was large enough to ensure the correct structure was estimated; we refer to this as the Exclusive Lasso. The second estimate, the Thresholded Exclusive Lasso, chooses the regularization parameter $\lambda$ that minimizes the BIC and then thresholds in each group keeping the index with the largest magnitude. We also compare our method to competitors and logical extensions of competitors in the literature. We base three comparison methods on the Lasso: First, we take the largest regularization parameter that yields exactly five nonzero coefficients (Lasso); second, we take the largest $\lambda$ that has nonzero indices in each group and then threshold group-wise to keep the coefficient in each group with the largest magnitude (Thresholded Lasso); third, we take the first coefficient along the Lasso regularization path to enter the active set from each group (Thresholded Regularization Path). Our final two comparison methods use Marginal Regression: First, we take the five indices that maximize $|X_i^Ty|$ (Marginal Regression); second, we take the one coefficient in each group that maximizes $|X_{i}^Ty|$ for $i\in g$ ( Group-wise Marginal Regression). For all methods we select a set of variables $S$, and then use the data matrix restricted to this set $X_S$ to calculate an Ordinary Least Square estimate $\betah_S$. The prediction error is calculated using $\betah_S$. Results in terms of prediction error and variable selection recovery are given in Table 1.

\begin{table*}[h1]\centering
\tiny
\ra{1.3}
\begin{tabular}{@{}rrrrrrrrr@{}}\toprule

& & Exclusive & Lasso& Marginal& Group-wise & Thresholded & Thresholded&  Thresholded\\
& & Lasso & & Regression& Marginal & Exclusive  & Lasso & Regularization \\
& & & & & Regression& Lasso &  &  Path\\  \midrule
$\text{w=.9, b=.9}$\\
&\text{True Vars} &  2.180   (1.02) & 2.160   (0.82) &  1.340 (0.63) & 1.500  (0.84)  &  3.760  (0.96) & 1.760 (1.06)   & 2.080 (0.97)        \\
&\text{False Vars}  & 2.820 (1.02) & 2.840 (0.82) &  3.660 (0.63) & 3.500 (0.84) & 1.240 (0.96) & 3.240 (1.06) & 2.920 (0.97)\\
&\text{Pred Er}  & 1.351 (0.15) & 1.433 (0.13) & 1.608 (0.14)&  1.411 (0.12) & {\bf 1.115} (0.13) &  1.411 (0.17) & 1.325 (0.15)\\
$\text{w=.9, b=.6}$\\
&\text{True Vars} &  3.86 (0.88) &  3.700   (0.81) & 2.10   (0.74) &  4.020 (0.82) &   4.480  (0.68)  &  4.060   (1.10)   & 3.96 (0.90)        \\
&\text{False Vars}  &1.14 (0.88) & 1.300 (0.81) &  2.90 (0.74) & 0.980 (0.82) & 0.520 (0.68) & 0.940 (1.10) &  1.04 (0.90) \\
&\text{Pred Err}  & 1.11 (0.10) & 1.236 (0.17) & 1.55 (0.16) & 1.102 (0.11) & {\bf 1.064} (0.09) & 1.129 (0.15) & 1.10 (0.11)\\
$\text{w=.6, b=.6}$\\
&\text{True Vars} &  4.720 (0.50)  &  4.600  (0.53)  &  3.620 (0.53) &  4.200 (0.49)    & 4.940  (0.24)  & 4.720  (0.45) &  4.740 (0.44)         \\
&\text{False Vars}  & 0.280 (0.50) & 0.400 (0.53) & 1.380 (0.53) & 0.800 (0.49) & 0.060 (0.24) & 0.280 (0.45) & 0.260 (0.44) \\
&\text{Pred Err}  & 1.066 (0.15) & 1.094 (0.15) & 1.304 (0.15) & 1.162 ( 0.15) & {\bf 1.022} (0.10) & 1.062 (0.13) &  1.057 (0.13)\\
\bottomrule
\end{tabular}

\caption{We compare the Exclusive Lasso and a thresholded version of the Exclusive Lasso to alternative variable selection methods as described in the Simulation section.  Here, there is one nonzero coefficient in each of the five groups, $n=100$ and $p=100$, and we vary the amount of between (b) and within (w) group correlation of the design matrix with Toeplitz covariance.  The Thresholded Exclusive Lasso outperforms all of the competing methods in both the recovery of truly nonzero variables and prediction error. }
\end{table*}

The thresholded version of the Exclusive Lasso outperforms all other methods at all levels of correlation, likely because it selects more variables that are truly nonzero. We observe that the thresholded estimators generally perform better then the non thresholded estimators. Among non-thresholded estimators, the Exclusive Lasso also performs the best at all levels of correlation. These simulations highlight the Exclusive Lasso's robustness to moderate and large amounts of correlation, which is important considering we expect variables in the same group to be similar and possibly highly correlated with each other.

In the second set of simulations, we also simulate data using the model $y = X\beta^* + \epsilon$ where $\epsilon \sim N(0,1)$ and $\beta^*$ is the true parameter for $n=p=100$. In these simulations the variables are divided into the same five equal-sized groups but the true parameter can be nonzero at more then one index in each group. Specifically, there are seven nonzero coefficients distributed so that three groups have exactly one nonzero index and two groups have two nonzero indices each.  We simulate the design matrices in the same way we simulate design matrices in the first set of simulations to have varying levels of between and within group correlation. 

We compare three methods: the Exclusive Lasso, the Lasso, and the Lasso applied independently to each group. For all methods, we use the BIC to select the regularization parameter. When we apply the Lasso separately to each group we use separate regularization parameters as well. 

Results in terms of prediction error and variable selection given in Table 2.

\begin{table*}[h1] \centering
\footnotesize
\ra{1.3}
\begin{tabular}{@{}rrrrr@{}}\toprule

& & $\text{Exclusive}$ & $\text{Lasso}$& $\text{Group-wise Lasso}$\\ \midrule
$\text{w=.9, b=.9}$\\
&\text{True Vars} &  6.820 (0.48) & 6.940 (0.24) & 4.920 (0.88)      \\
&\text{False Vars}  &6.280 (2.41) & 9.380 (3.08) & 5.880 (1.86)\\
&\text{Pred Er}  & {\bf 1.262} (0.22)&   1.295 (0.22)& 1.967 (0.64)\\
$\text{w=.9, b=.6}$\\
&\text{True Vars} & 6.740 (0.69) &  6.940 (0.24) &  4.780 (1.02)    \\
&\text{False Vars}  &6.420 (2.64) & 9.360 (3.35) &  6.180 (2.03) \\
&\text{Pred Err}  & {\bf 1.232} (0.22) & 1.259 (0.23) & 1.944 (0.54)\\
$\text{w=.6, b=.6}$\\

&\text{True Vars} & 7.000 (0.00) & 7.000 (0.00)   &   6.720 (0.45)     \\
&\text{False Vars}  &3.940 (2.61) &   5.320 (3.80) & 2.080 (1.28) \\
&\text{Pred Err}  &{\bf  1.197 }(0.19)& 1.233 (0.21) & 1.265 (0.29)  \\

\bottomrule
\end{tabular}

\caption{We compare the Exclusive Lasso to the Lasso and the Group-wise Lasso with BIC model selection for the second simulation scenario where we have five groups with either one or two true variables per group for a total of seven true variables. Again, $n=100$ and $p =100$ with the amount of between and within group correlation of the design matrix is varied.  The Exclusive Lasso performs best in terms of variable selection and prediction error.}
\end{table*}

The Exclusive Lasso has the best prediction error across all three simulations. The Exclusive Lasso selects fewer false variables then the Lasso and selects more true variables then the Group-wise Lasso. These simulations also suggest the Exclusive Lasso is more robust to high levels of correlation. Overall, our results suggest that the Exclusive Lasso performs best at within group variable selection when we have known group structure with relatively large amounts of correlation within and or between groups.

\section{NMR Spectroscopy Study}
\label{s:nmr}

Finally, we illustrate an application of the Exclusive Lasso for selecting the chemical shift of molecules in Nuclear Magnetic Resonance (NMR) spectroscopy. NMR spectroscopy is a high-throughput technology used to study the complete metabolic profile of a biological sample by measuring a molecule's interaction with an external magnetic field \citep{de2013vivo,cavanagh1995protein}. This technology produces a spectrum where the chemical components of each molecule resonate at a particular ppm. See Figure 4.b for example. A central analysis goal of NMR spectroscopy is identifying and quantifying the molecules in a given biological sample. This is challenging for numerous reasons discussed in \citep{ebbels2011processing,weljie2006targeted,zhang2009interdependence}. We seek to use the Exclusive Lasso to solve one of the major analysis challenges with NMR spectroscopy: accounting for positional uncertainty when quantifying relative concentrations of known molecules in a sample. Known as ``chemical shifts'', every molecules' chemical signature is subject to a random translation in ppm (Figure 4.a) due to the external physical environment of the sample \citep{de2013vivo} . One way to model this positional uncertainty, is to create an expanded dictionary of shifted molecules to use for quantification. With this expanded dictionary, we can consider each molecule and its shifts as a group, and use the Exclusive Lasso to select the best shift of each molecule for quantification. 

We choose not to use real NMR spectroscopy data as often true molecules and true concentrations are unknown. Instead we create a simulation based on real NMR molecule spectra in order to test our method for the purpose of NMR quantification. In our application, we simulate an NMR signal using a dictionary of reference measurements for thirty-three unique molecules. The dictionary, $X \in \mathbb{R}_+^{4000 \times (33*11)}$, consists of spectra for thirty-three molecules and ten artificial positional shifts for each molecule, five left and five right. These shifts are no more then .05ppm greater than or less than the reference measurement yielding eleven possible positions for each molecule. We use one randomly selected shift for each molecule, hence simulating the positional uncertainty found in real data. The columns of this expanded dictionary are strongly correlated with each other. Molecules are correlated with their ten shifts as well as other molecules with similar chemical structures. If we consider each molecule and its shifts a group, this results in a data set that has high correlation between groups as well as high correlation within each group as seen in Figure 5.a.

The simulated NMR signal, $y$, is a linear combination of the molecules in the dictionary with values chosen so that the signal has several properties that we observe in real data. For example, real NMR data can contain several unique molecules. Many of these will resonate at similar frequencies, causing peaks to overlap \citep{de2013vivo}. Informally, this yields signals that appear smoother with less pronounced peaks because of the crowding. With thirty-three molecules we can recreate this effect in the region between .5 and 0 ppm  (see Figure 5.b). We then simulate our signal using positive noise so that $y= X\beta^* + \epsilon$ where $\epsilon$ is the absolute value of Gaussian noise; this is done as real NMR spectra is non-negative.

\begin{figure}
        \centering
        \begin{subfigure}[b]{0.3\textwidth}
                \includegraphics[width=\textwidth]{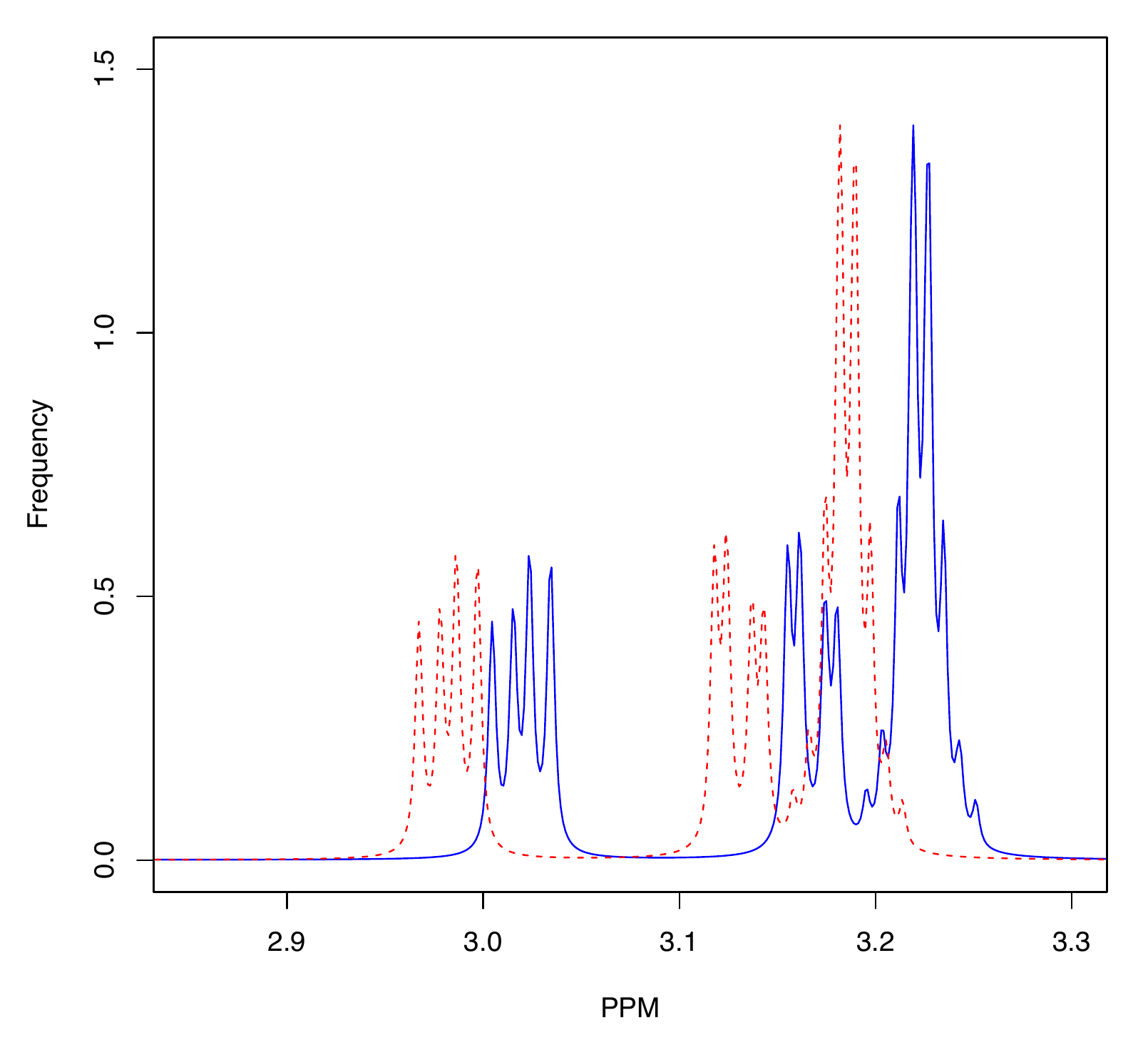}
                \caption{}
                \label{fig:shift}
        \end{subfigure}%
        ~ 
        \begin{subfigure}[b]{0.7\textwidth}
                \includegraphics[width=\textwidth]{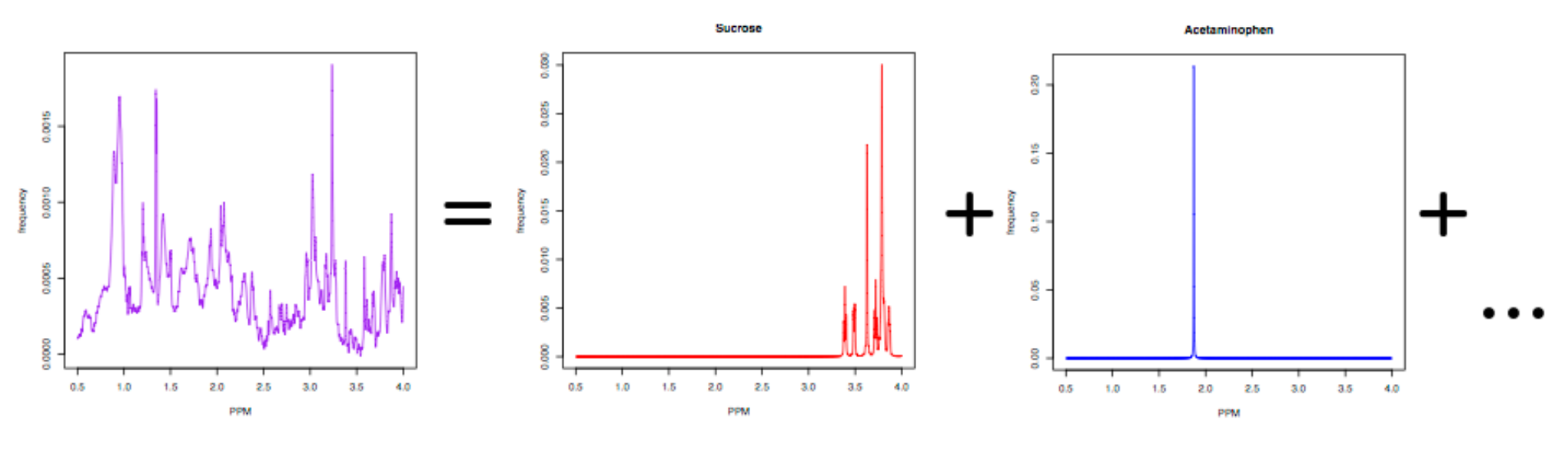}
                \caption{}
                \label{fig:nmr}
        \end{subfigure}
            \caption{ (a) Positional uncertainty of the chemical shift for the molecule Carnosine. All NMR spectroscopy signals are subject to random translations in ppm, due to the chemical environment of the sample. (b) NMR spectra of a neuron cell sample. NMR spectroscopy measures concentrations of all molecules in a sample. The observed signal is a linear combination of its unobserved component molecule's chemical signatures.}\label{fig:nor}
\end{figure}

We then use each method, the Exclusive Lasso, the Lasso, and the Group-wise Lasso, to select a set of variables $S$, consisting of one shift from each molecules' group of chemical shifts. Where applicable we use the thresholded versions of the estimates where we select $\lambda$ using the BIC and threshold group-wise so that there is only one nonzero variable in each group. Finally, we compare these methods to an ordinary least squares estimate that uses the original un-expanded dictionary without modeling the positional shifts. In Table 3, we report the prediction error and mean squared error, $MSE = \frac{1}{p} \underset{i=1}{\overset{p}{\sum}}(\beta_i^* - [(X_S^TX_S)^{-1}X^T_Sy]_i)^2$ so that we can accurately compare the methods as variable selection procedures. This measure eliminates the shrinkage that occurs with penalized regression methods and allows us to focus on how accurately we recover the concentrations of each molecule.

\begin{table*} \centering
 \begin{tabular}{rrr}
 \hline
 & Mean Squared Error $(\beta)$  & Prediction Error\\
 \hline
 Exclusive Lasso & 1.072(.03) & 1.339e-04(9.797e-07)\\
 OLS regression & 2.871(.06) & 2.605e-04(1.162e-06)\\
Marginal Regression &  1.163(.23) &  1.452e-04(1.841e-05)\\
Lasso &  2.092(.14) &   8.025e-05( 1.091e-05)\\
\hline

\end{tabular}
\caption{In our simulation using NMR spectroscopy data, we seek to quantify concentrations of molecules in a sample (see $MSE(\beta)$) under positional uncertainty in chemical shifts. Here, OLS regression quantifies concentrations without accounting for positional uncertainty whereas the Exclusive Lasso, Marginal Regression and the Lasso account for positional uncertainty by selecting one chemical shift for each molecule from an expanded dictionary. Given the selected variables, $\hat{S}$, these methods use OLS estimates for $X_{\hat{S}}$ to estimate $\hat{\beta}$ and quantify concentrations, the accuracy of which is measured by $MSE(\beta) = \frac{1}{p} \|\hat{\beta} - \beta^*\|_2^2$.  }
\end{table*}

\begin{figure}
        \centering
        \begin{subfigure}[b]{0.38\textwidth}
                \includegraphics[width=\textwidth]{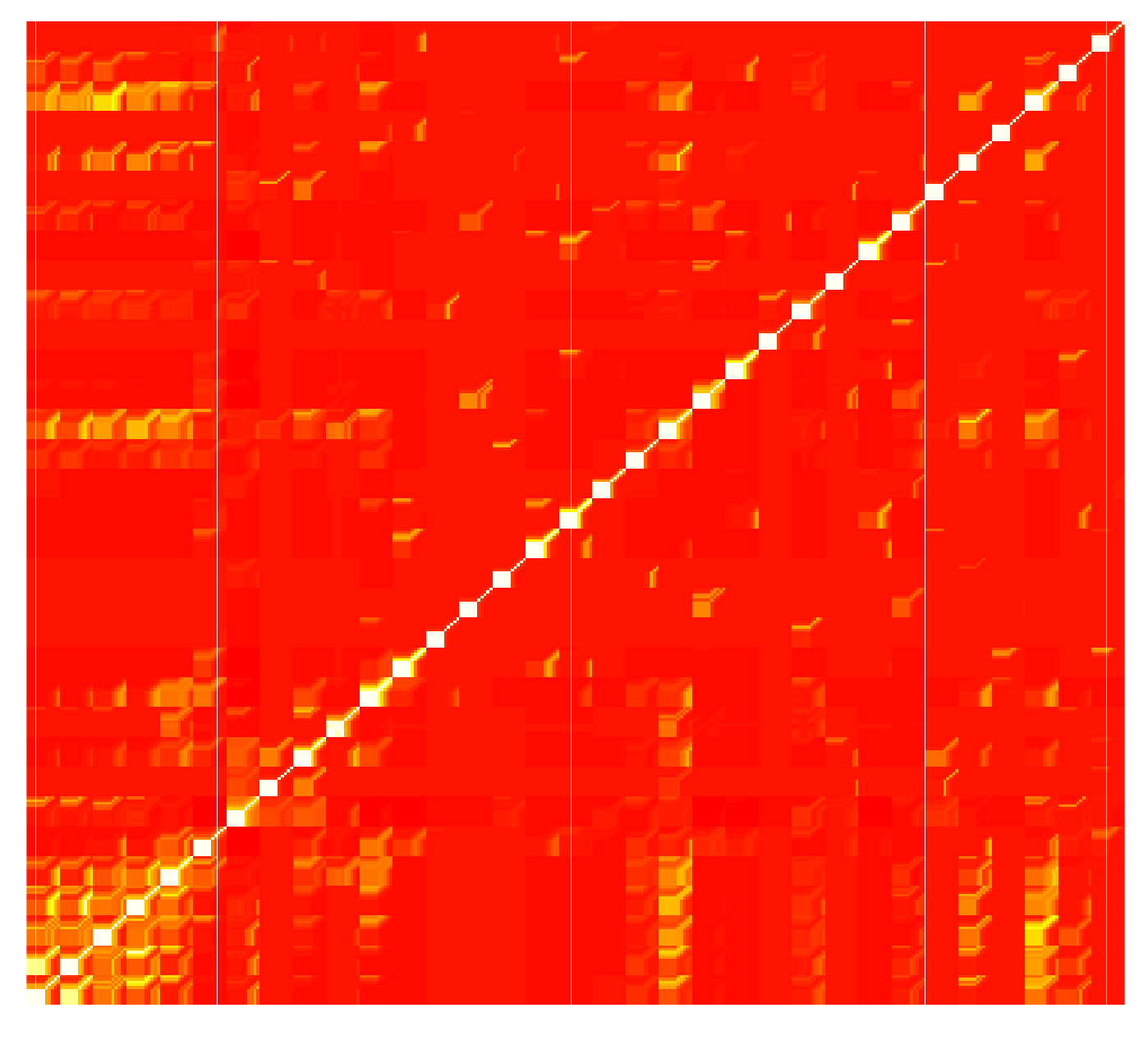}
                \caption{}
                \label{fig:cov}
        \end{subfigure}%
        ~ 
        \begin{subfigure}[b]{0.62\textwidth}
                \includegraphics[width=\textwidth]{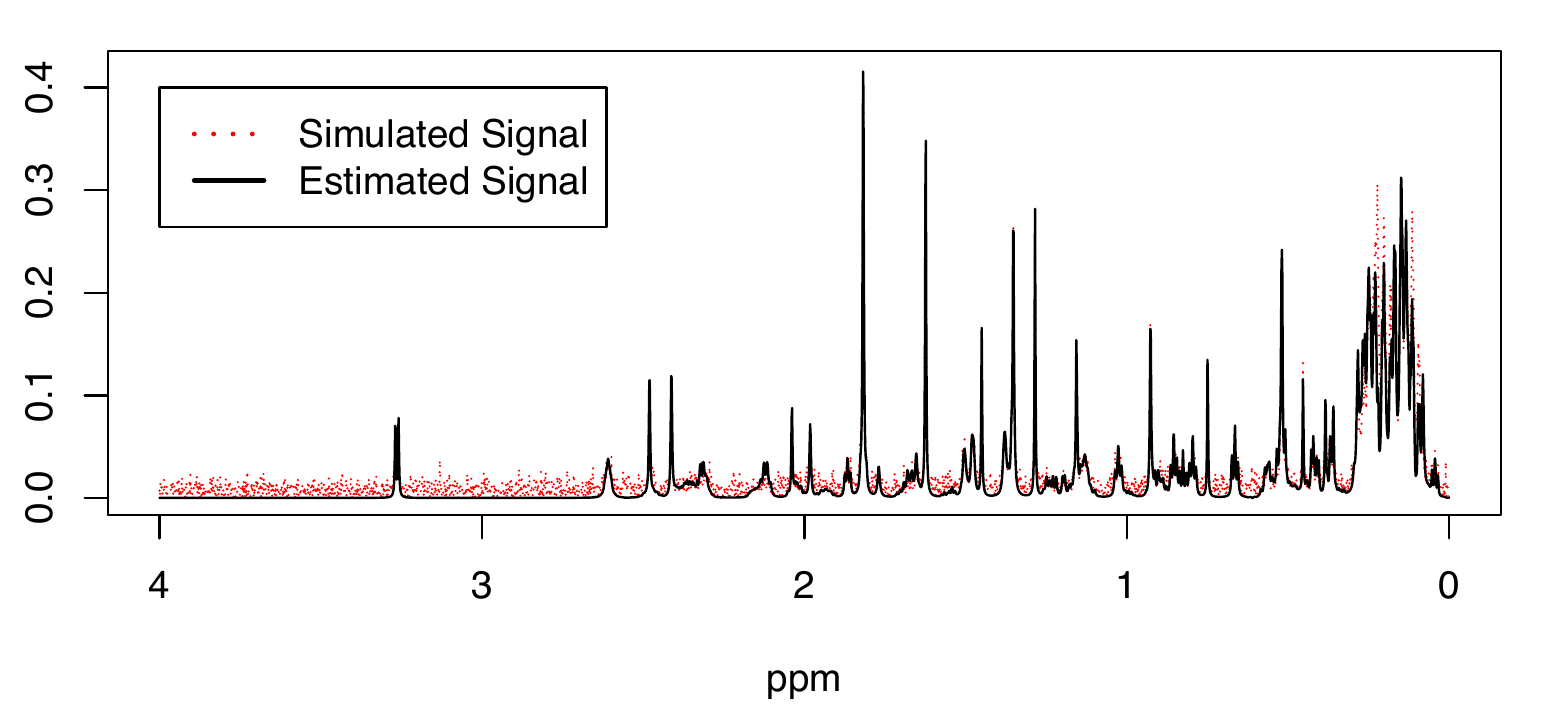}
                \caption{}
                \label{fig:quant}
        \end{subfigure}
            \caption{(a) The covariance matrix for the expanded dictionary of molecules. We simulate chemical shifts by generating 10 lagged variables for each of the 33 molecules (blocks on the diagonal). A molecule and its 10 shifts comprise a group where each variable in the group is very correlated with every other member in the group. We can also see that the molecules are very correlated with each other as we include molecules that are chemically similar. (b) The simulated NMR signal and the signal estimated using the Exclusive Lasso. The estimate recovers most of the peaks suggesting it is selecting a useful set of shifts. The estimate also zeros out most of the noise in the simulated signal. }\label{fig:nmr2}
\end{figure}

Among all methods, the Exclusive Lasso performs best at quantifying molecule concentrations under positional uncertainty.  This case study highlights a real example where there is high correlation both within and between pre-defined groups.  Consistent with our simulation studies, the Exclusive Lasso performs best in these situations.

\section{Discussion}
\label{s:dis}

Although others have introduced the Exclusive Lasso penalty, we are the first to investigate the method's statistical properties in the context of sparse regression for within group variable selection. We propose two new characterizations of the Exclusive Lasso in an effort to understand the estimate. The first characterization is an explicit definition of $\betah$ in terms of the support set that allows us to derive the degrees of freedom. This expression is similar to that of the ridge regression estimate, especially when there is exactly one nonzero variable in each group. The second characterization allows us to explore the properties of the active set.  We then prove that the Exclusive Lasso is prediction consistent under weak assumptions, the first such result. Additionally, we develop a new algorithm for fitting the Exclusive Lasso based on proximal gradient descent and derive the degrees of freedom so that we can use the BIC formula or the EBIC formula for model selection. 

Overall, we find that the Exclusive Lasso compares favorably to existing methods. Even though the Exclusive Lasso is a more complex composite penalty, convergence results for the Exclusive Lasso Algorithm are comparable to convergence rates for standard first order methods for computing the Lasso. Additionally, through several simulations, we find that the Exclusive Lasso not only selects at least one variable per group better then any existing method, but it also performs better when there is strong correlation both within groups and between groups. 

In this work, we focus on statistical questions important to the practitioner, but there are several directions for future work. Investigating variable selection consistency, overlapping or hierarchical group structures, and inference are important open questions. One could also use the Exclusive Lasso penalty with other loss functions such as that of generalized linear models. Additionally, there are many possible applications of our method besides NMR spectroscopy such as creating index funds in finance, and selecting genes from functional groups or pathways, among others.

Overall, the Exclusive Lasso is an effective method for within group variable selection in sparse regression; an R-package will be made available for others to utilize our method. 

\section{Appendix}

\section*{Proof of theorems 1 and 2}
The proof of theorems 1 and 2 follows the proof technique presented in \citet{chatterjee2013assumptionless}. There are several differences due to the structure of our penalty, however, the assumptions are the same. We assume that the columns of the design matrix $\{X_1 \dots X_p\}$ are possibly dependent random variables such that the covariance matrix for $\{X_1 \dots X_p\}$ is $\Sigma$. We assume the entries of $X$ are bounded so that $|X_{i,j}| \le M $ and that the data we observe $(Y_1,X_1) \dots (Y_n, X_n)$ is independent and identically distributed. We also assume the value of the penalty evaluated at the true parameter is bounded so that $P(\beta^*) \le K $ and that the response is generated by the linear model $Y = X\beta^* + \epsilon$ where $\epsilon \sim N(0,\sigma^2)$. Let $\mathcal{G}$ be a collection of predefined non overlapping groups such that $\underset{g \in \mathcal{G}}{\cup} g = \{1\dots p \}$.

Instead of the Exclusive Lasso penalty, we work with the equivalent constrained optimization problem $$\betah = \underset{\beta: P(\beta)\le K}{\text{argmin}}\|Y - X\beta\|_2$$

Let $C = \{X\beta: P(\beta) \le K\}$. By definition, $\hat{Y}$ is the projection of $Y$ onto the set $C$. For constrained optimization problems first order necessary conditions for an optimal solution state that for all $d$ in the linear tangent cone a solution to the problem $x^*$ necessarily satisfies
$f'(x^*;d)\ge 0$. In our case the linear tangent cone is the set $T_{\ell}(\hat{Y}) = \{(x-\hat{Y}): x \in C\}$ so an optimal solution satisfies $\langle -(Y-\hat{Y}), (x-\hat{Y}) \rangle \ge 0$ for all $x \in C$. Letting $x=Y^*$ we can rewrite $\langle (Y-\hat{Y}), (Y^*-\hat{Y}) \rangle \le 0$ as the  inequality

\begin{equation*}
\begin{aligned}
\|Y^* - \hat{Y}\|^2_2 &\le \langle (Y-Y^*), (\hat{Y}-Y^*)\rangle \\
&=\underset{i=1}{\overset{n}{\sum}}\epsilon_i \left(\underset{j=1}{\overset{p}{\sum}} (\betah_j - \beta_j^*)X_{i,j}\right)\\
&=\underset{j=1}{\overset{p}{\sum}} (\betah_j - \beta_j^*) \left(\underset{i=1}{\overset{n}{\sum}}\epsilon_i X_{i,j}\right)
\end{aligned}
\end{equation*}

Our assumption $P(\beta^*) \le K$ and the definition of $\betah$ let us bound $\underset{j=1}{\overset{p}{\sum}} (\betah_j - \beta_j^*)$ so that 

$$\underset{j=1}{\overset{p}{\sum}} (\betah_j - \beta_j^*) \le 2 (K + |\mathcal{G}|)$$

This implies that if we let $U_j = \underset{i=1}{\overset{n}{\sum}} \epsilon_i X_{i,j}$ then 

$$\|Y^* - \hat{Y}\|^2 \le  2 (K + |\mathcal{G}|)\underset{1\le j\le p}{\max}|U_j|$$

Because  $U_j \sim N\left(0,\sigma^2\underset{i=1}{\overset{n}{\sum}}X^2_{i,j}\right)$ we have the bound 

$$\mathbb{E}(\underset{1 \le j \le p}{\max}|U_j| ) \le M\sigma \sqrt{2n\log(2p)}$$
See lemma 3 in  \cite{chatterjee2013assumptionless} for proof of the bound.
Therefore

$$\mathbb{E}\|Y^* - \hat{Y}\|^2 \le  2 (K + |\mathcal{G}|) M\sigma \sqrt{2n\log(2p)}$$

which gives us theorem 2: 

$$\mathbb{E}[\widehat{MSPE}(\betah)] \le 2 (K + |\mathcal{G}|) M\sigma \sqrt{\frac{2\log(2p)}{n}} $$

We use this result to prove theorem 1. By the independence of the data $(Y,X)$ and $\betah$ we have 

$$\mathbb{E}(Y^* - \hat{Y})^2 =\underset{j,k=1}{\overset{p}{\sum}}(\beta^*_j - \betah_j)(\beta^*_k - \betah_k)\mathbb{E}(X_jX_k)$$

note that 

$$\frac{1}{n}\|Y^* - \hat{Y}\|^2 = \underset{j,k=1}{\overset{p}{\sum}}(\beta^*_j - \betah_j)(\beta^*_k - \betah_k)X_jX_k$$

Combining these two expressions yields

$$\mathbb{E}(Y^* - \hat{Y})^2 -\frac{1}{n}\|Y^* - \hat{Y}\|^2 = \underset{j,k=1}{\overset{p}{\sum}}(\beta^*_j - \betah_j)(\beta^*_k - \betah_k) [\mathbb{E}(X_jX_k) - \frac{1}{n}X_jX_k]$$

We then define $V_{j,k} = [\mathbb{E}(X_jX_k) - \frac{1}{n}X_jX_k]$ and note that it is bounded $|V_{j,k}| \le 2M^2$. By Hoeffding's inequality 

$$\mathbb{E}(\underset{1\le j,k \le p}{\max}|V_{j,k}|) \le 2M^2\sqrt{\frac{2\log(2p^2)}{n}}$$

We use a version of Hoeffding's inequality that is rather uncommon so we refer the interested reader to the appendix of \cite{chatterjee2013assumptionless} for a derivation of the result.

Finally 

$$\mathbb{E}(Y^* - \hat{Y})^2 -\frac{1}{n}\|Y^* - \hat{Y}\|^2 \le 4(K + |\mathcal{G}|)^2 \underset{1\le j,k \le p}{\max}|V_{j,k}|$$

Combining our results yields theorem 1

$$\mathbb{E}(Y^* - \hat{Y})^2 \le 2(K + |\mathcal{G}|)M\sigma \sqrt{\frac{2\log(2p)}{n}} + 8(K + |\mathcal{G}|)^2M^2\sqrt{\frac{2\log(2p^2)}{n}}$$

\section*{Proof of corollary 1}
The $MSPE(\betah)$ is equal to $\mathbb{E}\|\betah - \beta^*\|_{\Sigma}$. We can bound $\|\betah - \beta^*\|_2$ by the MSPE such that $\|\betah - \beta^*\|^2_2 \le \frac{1}{c}\|\betah - \beta^*\|_{\Sigma}^2$ showing that $\|\betah - \beta^*\|_2^2$ goes to 0 as $MSPE(\betah)$ goes to 0.

\section*{Proof of theorem 3}
Our coordinate descent algorithm calculates the proximal operator by solving the optimization problem

$$\text{prox}_P (y) = \underset{x}{\text{argmin}} \frac{1}{2} \|y - x\|_2^2 + \lambda P(x)$$

We show that the assumptions for theorem 4.1 from \cite{tseng2001convergence} hold for the problem above.  For a function of the form 

$$f(x) = g(x) + h(x)$$

where $g$ is convex and differentiable and $h$ is convex but not necessarily differentiable, verifying the assumptions involves showing that 
\begin{enumerate}
\item The differential part of our function $g$ satisfies assumption (A1) from \cite{tseng2001convergence} 

$$\text{Assumption: (A1) The domain of } g \text{ is open and } g \text{ is Gateux differentiable}$$

\item The function $f$ is a regular function. \\
\item The level set $X_0 = \{x: f(x) \le f(x^0)\}$ is compact and that $f$ is continuous on $X_0$
\item For every pair $i,k \in \{1\dots p\}$ it follows that $f$ is jointly pseudo convex in $x_i$ and $x_k$
\end{enumerate}

First we state several definitions. 

We say direction $d$ is a vector in $\mathbb{R}^n$. We allow $d_k$ to be the scalar in the $k^{th}$ position in the vector $(0 \dots 0, d_k , 0 \dots 0)$. We abuse notation if the meaning is unambiguous, and also let $d_k$ denote the entire vector with $0$s in all positions except for the $k^{th}$ position.
It is typical to define first order optimality conditions in terms of the Gateaux derivative. We however use the more general forward variation defined as follows:
\begin{definition}
 For a function $f$ the forward variation in direction $d$ at x is $$f_+'(x;d)=\underset{t \downarrow 0}{\lim} \frac{f(x+td)-f(x)}{t}$$
\end{definition}
The Gateaux derivative exists if both the forward and backward variation exist and are equal. Tseng uses the Gateaux derivative to define his optimality conditions but for our unconstrained convex non-differentiable problem it is necessary and sufficient for a minimizer of $f$ to satisfy $f_+' (x;d) \ge 0 $ for all $d \in \mathbb{R}^n$. We also use a notion called regularity. Note that this is the same definition of regularity given in \cite{tseng2001convergence} communicated here for convenience. Throughout the rest of the paper we use the forward variation and the directional derivative interchangeably.

\begin{definition}
A function f is regular at $x$ if $f'(x;d) \ge 0 $ for all $d$ such that $f'(x;d_k) \ge 0$
\end{definition}

Regularity ensures that if we have a point that minimizes f coordinstewise, then the point minimizes the function f.

\begin{definition}
A function $f$ is pseudoconvex if $f(x + d) \ge f(x)$ whenever $x \in dom(f)$ and $f'(x;d) \ge 0$
\end{definition}

Assumption 1: The differential part of our function $g$ satisfies assumption (A1) from \cite{tseng2001convergence} 

\begin{proof}
If we let 
$$g(x) =  \frac{1}{2} \|y - x\|_2^2$$

 its domain is $\mathbb{R}^n$ which is an open set. We must also show that $g(x) = \frac{1}{2}\| y - x\|_2^2$  is Gateux-differntiable on $\mathbb{R}^n$.
 \begin{equation*}
 \begin{aligned}
g'(x; d) &= \underset{t \downarrow 0}{\lim} \frac{g(x + td) - g(x)}{t} \\
&= \underset{t \downarrow 0}{\lim} \frac{1}{2t} \|y-(x+td)\|_2^2 - \frac{1}{2t} \|y-x\|_2^2  \\
&= -(y - x)^Td \\
&= \nabla g(x)^Td
\end{aligned}
\end{equation*}
A similar argument holds as $t \uparrow 0$
\end{proof}

Assumption 2: the function $f$ is a regular function

\begin{proof}
 Our goal is to show that if we have a point $x$ that minimizes $f$ point wise i.e. that $f'(x;d_k) \ge 0$ for all $d_k$ then we have a point that minimizes $f$ and satisfies the standard first order necessary and sufficient condition for optimality $f'(x;d) \ge 0$ for all $d \in \mathbb{R}^n$. We know that $g(x) = \frac{1}{2}\| y - x\|_2^2$  is Gateux-differntiable on $\mathbb{R}^n$.

Next we show that the entire function $f (x) = g(x) + h(x)$ is regular. Assume that the point $x$ minimizes $f$ point wise therefore satisfying:
$$f'(x;(0...0,d_k,0...0)) \ge 0$$

for all $d_k$. Then it follows that

\begin{equation*}
\begin{aligned}
f'(x;d) &= \nabla g(x)^Td + \underset{t \downarrow 0}{\lim} \frac{(\underset{i=1}{\overset{n}{\sum}} |x_i + td_i|)^2 - (\underset{i=1}{\overset{n}{\sum}}|x_i|)^2}{t} \\
&= \nabla g(x)^Td + \underset{t \downarrow 0}{\lim} \frac{\left (\underset{i=1}{\overset{n}{\sum}} |x_i + td_i| - \underset{i=1}{\overset{n}{\sum}}|x_i|\right)\left(\underset{i=1}{\overset{n}{\sum}} |x_i + td_i| + \underset{i=1}{\overset{n}{\sum}}|x_i| \right)}{t} \\
&= \nabla g(x)^Td + \underset{t \downarrow 0}{\lim} \frac{\left(\underset{i=1}{\overset{n}{\sum}} |x_i + td_i| - \underset{i=1}{\overset{n}{\sum}}|x_i|\right)}{t}\underset{t \downarrow 0}{\lim}\left(\underset{i=1}{\overset{n}{\sum}} |x_i + td_i| + \underset{i=1}{\overset{n}{\sum}}|x_i|\right) \\
&= \nabla g(x)^Td + \underset{t \downarrow 0}{\lim} \frac{\underset{i=1}{\overset{n}{\sum}} |x_i + td_i| - \underset{i=1}{\overset{n}{\sum}}|x_i|}{t}2\|x\| \\
&\ge \nabla g(x)^Td + \underset{i=1}{\overset{n}{\sum}} \underset{t \downarrow 0}{\lim}\frac{ |x_i + td_i| - |x_i|}{t}2\|x\| \\
&=  \underset{i=1}{\overset{n}{\sum}} f'(x;(0,\dots,0,d_k,0,\dots,0))\\
&\ge 0
\end{aligned}
\end{equation*}
\end{proof}

Assumption 3: The level set $X_0 = \{x: f(x) \le f(x^0)\}$ is compact and that $f$ is continuous on $X_0$
\begin{proof}
We show that the function is continuous by showing that the penalty is continuous and that the differentiable part of the objective function is continuous.  Let $x,y \in X_0$ then there exists a $\delta$ such that for 

$$|x-y| \le \delta$$

it follows that 

$$|P(x) - P(y)| \le \epsilon$$

To find $\delta$ consider

\begin{equation*}
\begin{aligned}
|P(x) - P(y)| & \le P(x-y) \\
&= \underset{g \in \mathcal{G}}{\sum}(\underset{i \in g}{\sum} |x_i - y_i|)^2\\
&\le \underset{g \in \mathcal{G}}{\sum}(\underset{i \in g}{\sum} \delta_i)^2
\end{aligned}
\end{equation*}

Note that the first line follows from the reverse triangle inequality. If $i \in g$ then for any $\epsilon >0$ we can define $\delta$ such that $\delta_i = \frac{\sqrt{\epsilon}}{n_g \sqrt{|\mathcal{G}|}}$ which shows that the penalty is continuous on the set. 

To show that the term $\|y-x\|_2^2$ is continuous  consider two points  $x,z \in X_0$ and suppose 

$$|x-z| \le \delta$$

Consider 

\begin{equation*}
\begin{aligned}
\left|\|y-x\|^2_2 - \|y-z\|^2_2\right| &\le \|(y-x)-(y-z)\|^2_2\\
&= \|x-z\|^2_2\\
&\le (\|x-z\|_1)^2\\
&=( \underset{i}{\sum}|x_i - z_i|)^2\\
 &\le ( \underset{i}{\sum} \delta_i)^2
\end{aligned}
\end{equation*}

So for $\delta_i \le \frac{\sqrt{\epsilon}}{n}$ the term $\|y-x\|^2_2$ is continuous. Therefore $f$ is continuous because the sum of continuous functions is a continuous function. Using theorem 1.6 of \cite{rockafellar2009variational}, continuity implies that the level sets are closed.

 The level sets also must be bounded. For any level set

$$X_0 = \{x: \|y-x\|^2_2 + \lambda P(x) \le \|y-x_0\|^2_2 + \lambda P(x_0)\}$$

If we let $\|y-x_0\|^2_2 + \lambda P(x_0)=\alpha$ we can consider a vector of the form $x_{\alpha} = (0,\dots,0, \sqrt{\frac{|\alpha| +1}{\lambda}} ,0, \dots,0)$. Our penalty evaluated at this vector gives $\lambda P(x_{\alpha} ) = |\alpha| + 1 > \alpha $. Since $\|y - x\|\ge 0$ for all $x \in \mathbb{R}^n$ the objective function $f(x_{\alpha}) > \alpha$ . This implies that for all $x \in X_0$ there exists an $M \in \mathbb{R}$ such that $\underset{i}{\max}|x_i| \le M$. Therefore the level sets are bounded.

By the Heine-Borel theorem since $X_0$ a closed bounded subset of $\mathbb{R}^n$ it is compact.
\end{proof}

Assumption 4: For every pair $i,k \in \{1 \dots p\}$ it follows that $f$ is jointly pseudoconvex in $x_i$ and $x_k$.

\begin{proof}
For any pair of indices $i,k \in \{1 \dots p\}$ the function 

$$\|y-x\|^2_2 + \lambda P(x)$$

is jointly convex in $x_i$ and $x_k$. 
Suppose indices $i$ and $k$ are in the same group. We can rewrite the objective function as

\begin{equation*}
\begin{aligned}
f_1(x_i, x_k) &= \|x\|_2^2 - 2y^Tx + y^Ty + \lambda \underset{g \in \mathcal{G}}{\sum} (\underset{j \in g}{\sum}|x_j|)^2\\
&= x_i^2 + x_k^2 + x_i c_0 + x_kc_1 + (x_i + x_k)^2 + c_2
\end{aligned}
\end{equation*}

 where $c_0,c_1,c_2$ are terms constant in $x_i$ and $x_k$ and $y_{i,k} = (y_i,y_k)$ and $x_{i,k} = (x_i,x_k)$ are the vectors restricted to indices $i,k$. Both the $\ell_2$ norm and the affine function of $x_{i,k}$ are convex. The function $f_1(x_i,x_k) $ has a positive semidefinite hessian so it is also convex. 

If $i,k$ are in different groups we rewrite the objective function as 

\begin{equation*}
\begin{aligned}
f_2(x_i,x_k) &= 2x_i^2 +2 x_k^2 + c_0x_i + c_1x_k + c_2\\
\end{aligned}
\end{equation*}

Function $f_2$ also has a positive semidefinite hessian so it is also convex.  \\

Therefore the function $f$ is convex in every pair of indices which implies that it is pseudoconvex in every pair of indices. 

\end{proof}

Given that the objective function satisfies all of the assumptions for \cite{tseng2001convergence} Theorem 4.1 we can say that our coordinate descent algorithm converges to a stationary point. Because our function is convex the stationary point is a global minimum. 

\section*{Proof of theorem 4}

Our result depends on work by \cite{schmidt2011convergence}. We seek the convergence rate for the our Exclusive Lasso algorithm. In our algorithm at each step $k$ the proximal operator is computed to within a small error $\epsilon_k$ such that the iterate $x_k = \epsilon_k + \underset{x}{\argmin}\|y-x\|_2^2 + \lambda P(x)$.  As long as the sequence of errors is summable the algorithm will converge at a rate of at least $O(1/k)$ when the following assumptions hold. For function $f(x)= g(x) + h(x)$ we assume

\begin{enumerate}
\item The function $g$ is convex with a lipschitz-continuous gradient.
\item The function $h$ is a lower semi-continuous proper convex function.
\item There exists a point $x^* \in \mathbb{R}$ that minimizes $f$.
\item The points $x_k$ are $\epsilon_k$-optimal solutions to the proximal operator optimization problem at iteration $k$. 
\end{enumerate}

We must verify that these assumptions hold for the Exclusive Lasso \\

Assumption 1: In our case $g(\beta) = \frac{1}{2}\|y - X\beta\|_2^2$ so 
\begin{equation*}
\begin{aligned}
|\|y - X\beta_1\|_2 - \| y-X\beta_2\|_2| &\le \|(y-X\beta_1) - (y-X\beta_2)\|_2\\
&= \|X(\beta_1 - \beta_2)\|_2\\
&\le \|X\|_2\|(\beta_1 - \beta_2)\|_2\\
&= \lambda_{max}(X^TX)\|(\beta_1 - \beta_2)\|_2
\end{aligned}
\end{equation*}

which implies that $g$ is lipschitz- continuous with lipschitz constant $L = \lambda_{max}(X^TX)$ the largest eigenvalue of $X^TX$. \\

Assumption 2: Because $\|x\|_1$ is continuous for all $x \in \mathbb{R}^n$ and $b(z) = z^2$ is continuous for all $z \in \mathbb{R}$ their composition $\|x\|_1^2$ is continuous at all points in $\mathbb{R}^n$. To show that the penalty is convex we will consider the convexity of $ f(x) = \|x\|^2_1$.
For $t \in [0,1]$ and $x,z \in \mathbb{R}^n$

\begin{equation*}
\begin{aligned}
\|tx+ (1-t)z\|_1^2 &\le (t\|x\|_1 + (1-t)\|z\|)^2\\
& \le t\|x\|_1^2 + (1-t)\|z\|_1^2
\end{aligned}
\end{equation*}
Therefore $ f(x) = \|x\|^2_1$ is convex. The convexity of $P(\beta)$ follows from the fact that the sum of convex functions is also convex.

The penalty is proper by definition since for all $x \in \mathbb{R}^n$ we have $P(x) \ne \infty$ \\

Assumption 3: Using theorem 1.9 from \cite{rockafellar2009variational} we show existence of a solution.  We need the level sets  $X_{\alpha}= \{x: f(x) \le \alpha\}$ to be bounded for all $\alpha \in \mathbb{R}$. Consider a vector of the form $\betah_{\alpha} = (0,\dots,0, \sqrt{\frac{|\alpha| +1}{\lambda}} ,0, \dots,0)$. Our penalty evaluated at this vector gives $\lambda P(\beta_{\alpha} ) = |\alpha| + 1 > \alpha $. Since $\|y - X\beta\|\ge 0$ for all $\beta \in \mathbb{R}^n$ the objective function $f(\beta_{\alpha}) > \alpha$ . This implies that for all $x \in X_{\alpha} $ there exists an $M \in \mathbb{R}$ such that $\underset{i}{\max}|x_i| \le M$. Therefore the level sets are bounded. 

We have already shown that both $g$ and $h$ are continuous so their sum must also be continuous. Therefore because the level sets of our function $f$ are bounded, and $f$ is continuous and proper by theorem 1.9 there exists a minimum to our objective function $f$.

Assumption 4: This assumption holds by theorem 3.

Therefore by proposition 1 from \cite{schmidt2011convergence} the Exclusive Lasso algorithm converges at a rate of $O(1/k)$. 

\section*{Proof of theorem 5}
For a continuous and almost differentiable function $g$, Steins formula 

$$df (g) = \mathbb{E}[(\nabla * g)(y)]$$

defines the degrees of freedom for normal random variables in terms of the function $(\nabla * g)$. The function $(\nabla * g)$ known as the divergence is defined for $g : \mathbb{R}^n \rightarrow \mathbb{R}^n$ as

$$(\nabla * g)(y) = \underset{i=1}{\overset{n}{\sum}} \frac{\partial g_i}{\partial y_i} $$

To derive the degrees of freedom for the Exclusive Lasso problem we need to prove that the estimate is a continuous and almost differentiable function of $y$. Tibshirani provides a lemma stating that
\begin{lemma}
For a convex set $C \subset \mathbb{R}^n$ the projection map $P_C$ and the map $I- P_C$ are continuous and almost differentiable.
\end{lemma}
For proof see \cite{tibshirani2012degrees}.

\begin{lemma}  The estimate $X\betah = (I - P_C )y$ for the set

$$C=\{u \in \mathbb{R}^n :P^*(X^Tu) \le \alpha \} $$

where

$$P^*(\beta) = \sqrt{\underset{g \in \mathcal{G}}{\sum} \|\beta_g\|^2_{\infty}}$$

is the dual norm of the square root of our penalty and $\alpha$ is a constant.

\end{lemma}

\begin{proof}

The dual norm of a norm $\|z\|$ is defined as the norm $\|x\|^*$ such that $\|z\| = \text{sup}\{\langle x,z \rangle : \|x\|^* \le 1\}$. Note that for the square root of our penalty 

$$\sqrt{P(\betah)} = \left \langle \frac{sign(\betah)\|\betah_{g_i}\|_1}{\sqrt{\sum_g \|\betah\|_1^2}}, \betah \right \rangle$$

This means that our dual norm is the norm such that $P^*( \frac{sign(\betah)\|\betah_{g_i}\|_1}{\sqrt{\sum_g \|\betah\|_1^2}}) \le 1$ which holds for the norm 

$$P^*(\beta) = \sqrt{ \underset{g \in \mathcal{G}}{\sum} \|\beta_g\|_{\infty}^2}$$

We show that $\theta = y-X\betah$ is equal to the projection of $y$ onto the set $C$. 
The projection $\theta = P_C(y)$ can be characterized as a point $\theta$ satisfying the first order optimality conditions for the constrained optimization problem $\underset{\theta \in C}{\min}\|y-\theta\|_2^2$. The first order optimality conditions are 

\begin{equation*}
\begin{aligned}
f'(\theta;d) &\ge 0\\
\langle y - \theta, \theta - u \rangle &\ge 0
\end{aligned}
\end{equation*}

for all $u \in C$

We must verify that $f'(\theta;d) \ge 0$.  If we let $\theta = y - X \betah(y)$ then

\begin{align}
\langle y - \theta, \theta - u \rangle  &= \langle X\betah, y - X\betah - u \rangle \\
&= \langle X\betah,y-X\betah\rangle - \langle X^T u,\betah \rangle \\
&= \frac{\alpha}{2} \sqrt{P(\betah)} - \langle X^Tu, \betah \rangle \\
&=   \underset{P^*(w) \le \frac{\alpha}{2}}{\text{max}}\langle w, \betah \rangle - \langle X^Tu, \betah \rangle \\
&\ge 0
\end{align}

Line 3 follows from the fact that there exists a regularization parameter such that the necessary conditions for the Exclusive Lasso problem are exactly the same as the necessary conditions for the optimization problem that uses the square root of the Exclusive Lasso penalty. Notice that if we let $\alpha = 2\lambda P(\betah)^{\frac{1}{2}}$ then $\lambda\partial  P(\betah) = \alpha\partial  \sqrt{P(\betah)}$. This implies that $\betah$  necessarily satisfies 

$$-X^T(y-X\betah) +  \alpha \partial \sqrt{P(\betah)} = 0$$

Taking the inner product with $\betah$ yields 

$$(X\betah)^T(y-X\betah) = \frac{\alpha}{2}  \sqrt{P(\betah)}$$

Line 5 follows for the set $C = \{ u \in \mathbb{R}^n : P^*(X^Tu) \le \frac{\alpha}{2} \}$ proving that $y-X\betah$ is equal to the projection of $y$ onto the set $C$. This implies that $X\betah = (I - P_C)y$

\end{proof}

Combining Lemmas 1 and 2 yields that the exclusive lasso estimate is continuous and almost differentiable. Next we define $\betah$ in terms of the support set $S$. First recall the KKT conditions 

$$-X^T(y - X\betah) + \lambda z = 0$$

where

$$z_i =  \left\{
  \begin{array}{lr}
    sign(\betah_i)\|\betah_g\|_1  \    &: \betah_i \ne 0,i \in g \\
    \left[-\|\betah_g\|_1,\|\betah_g\|_1\right] &: \betah_i = 0
  \end{array}
\right.
$$

Note that we can rewrite the sub gradient for the indices $i \in g \cap S$. If we let $s_{g \cap S} = sign(\betah_{g \cap S})$

$$z_{g \cap S} = s_{g \cap S}s_{g \cap S}^T \betah_{g \cap S}$$

We can write the sub gradient over the indices of the support as

$$z_S = M_S\betah_S$$

where $M_S$ is a block diagonal matrix with the matrices $\{s_{g \cap S}s_{g \cap S}^T: g\in \mathcal{G}\}$ on the diagonal. 

We can rewrite the KKT conditions with respect to the support set

$$
-\begin{bmatrix}
X_S^T \\
X^T_{S^c}
\end{bmatrix}
\left (y - \begin{bmatrix} X_S \ X_{S^c}\end{bmatrix} \betah\right) +\lambda \begin{pmatrix}z_S \\ z_{S^c} \end{pmatrix} = 0
$$

This is equal to 

\begin{equation*}
\begin{aligned}
-X_S^Ty + X^T_SX_S\betah_S + \lambda z_S &= 0\\
-X_{S^c}^Ty + X_{S^c}^TX_S\betah_S + \lambda z_{S^c} &= 0
\end{aligned}
\end{equation*}

We then solve for $\betah_S$ using $z_S = M_S\betah_S$ yielding

$$\betah_S = (X^T_SX_S + \lambda M_S)^{\dagger}X_S^Ty $$

Note that we are relying on the fact that we have already proved the existence of a solution to the optimization problem in the proof for theorem 4. This gives us an estimate $\hat{y} = X_S(X^T_SX_S + \lambda M_S)^{\dagger}X_S^Ty $. The divergence is therefore $$(\nabla * X\betah)(y) =trace[X_S(X^T_SX_S + \lambda M_S)^{\dagger}X_S^T]$$ which is equal to the sum of the eigenvalues.

\section*{Penalty}

For specific values of $X$ and $y$ the Exclusive Lasso will select more than one variable per group for all values of the regularization parameter $\lambda$. This means that although the Exclusive Lasso is designed to select exactly one element per group we cannot guarantee the Exclusive Lasso will enforce the correct structure. Consider an example. Suppose we characterize the Exclusive Lasso estimate using the equicorrilation set. Recall the equicorrilation set 
$$\E = \left\{i: \frac{|X_i^T(y-X\betah)|}{\|\betah_g\|_1} = \lambda\right\}$$

If we let $s$ be a vector such that $s_i = sign(\betah_i)$ for $i \in \E$ and $\gamma$ be a vector such that $\gamma_i = \|\betah_{g_i}\|_1$ where $g_i$ is the group for an index $i \in \E$. Let $\bar{\gamma}$ be a vector such that $\bar{\gamma}_i = \|\betah_{g_i}\|_1 - |\betah_i|$ then we can solve for $\betah$. 

\begin{equation*}
\begin{aligned}
X_{\E}^T(y - X_{\E}\betah_{\E}) &= \lambda \gamma s\\
& =\lambda \bar{\gamma}s  +\lambda  \betah_{\E}\\
\betah_{\E} &= (X_{\E}^TX_{\E} + \lambda I )^{-1}[X^T_{\E}y-  \lambda \bar{\gamma} s]
\end{aligned}
\end{equation*}

 Let $X  = I_2$ and we let $y^T = (1,1)$ then because $X$ is orthonormal the estimate simplifies to 

$$\betah_{\E} = \frac{1}{1 + \lambda} y - \frac{\lambda}{1 + \lambda}\gamma' s$$

In this case $\betah_1 = \betah_2$ so the term $\frac{\lambda}{1 + \lambda}\gamma' s$ is going to shrink both indices equally for all $\lambda$. This prevents the estimate from selecting exactly one element in each group. 

We conjecture that conditions on $X$ and $y$ for this to occur can be formalized, but this is beyond the scope of this work. Intuitively, this behavior occurs when two or more variables get shrunken equally. As such, this behavior is relatively rare in practice. If it does occur and one variable per group is desired, we propose to use BIC to select $\lambda$ and apply group-wise t

\section{Acknowledgements}

This material is based upon work supported by the National Science Foundation Graduate Research Fellowship Program under Grant No.0940902. GA acknowledges support from NSF DMS-1264058 and DMS-1209017. The authors would like to thank Dr. Zhandong Liu for helpful conversations while preparing the manuscript and for his help acquiring the NMR spectroscopy data. Any opinions, findings, and conclusions or recommendations expressed in this material are those of the authors and do not necessarily reflect the views of the National Science Foundation.

\bibliographystyle{Chicago}

\bibliography{example}

\end{document}